\documentclass[conference, 11pt, onecolumn]{ieeeconf}

\IEEEoverridecommandlockouts
\usepackage[usenames]{color}
\usepackage{enumerate}
\usepackage{url}
\usepackage{subfig}
\usepackage{amsfonts,mathrsfs}
\usepackage{amssymb,amsmath}
\usepackage{verbatim}
\usepackage{acronym}
\usepackage{mathtools}
\usepackage{thmtools,thm-restate}
\usepackage{cite}
\usepackage{graphicx}
\usepackage{algorithm}
\usepackage[noend]{algpseudocode}

\def\fskip#1{}

\def\E{\mathbb{E}}
\def\P{\mathbb{P}}
\def\R{\mathbb{R}}
\def\N{\mathbb{N}}

\def\F{\mathcal{F}}
\def\H{\mathcal{H}}
\def\G{\mathcal{G}}
\def\V{\mathcal{V}}
\def\O{\mathcal{O}}
\def\CP{\mathcal{P}}

\def\hq{\hat{q}}

\newtheorem{theorem}{Theorem}

\newtheorem{assumption}{Assumption}

\newtheorem{definition}{Definition}

\newtheorem{lemma}{Lemma}

\newtheorem{proposition}{Proposition}

\def\la{\langle}
\def\ra{\rangle}

\def\argmax{\mathop{\rm argmax}}
\def\bs{\boldsymbol}

\begin{document}
\title{Scalable and Independent Learning of Nash Equilibrium Policies in
$n$-Player Stochastic Games with Unknown Independent Chains}
\author{\authorblockN{Tiancheng Qin and S. Rasoul Etesami}
 \authorblockA{\vspace{-1.5cm}
}
\thanks{*The authors are with the Department of Industrial and Systems Engineering, Department of Electrical and Computer Engineering, and Coordinated Science Lab, University of Illinois Urbana-Champaign, Urbana, IL, USA 61801. Email: (tq6, etesami1)@illinois.edu. This material is supported by the Air Force Office of Scientific Research under award number FA9550-23-1-0107 and the NSF
CAREER Award under Grant No. EPCN-1944403.}
}

\maketitle
\begin{abstract}
We study a subclass of $n$-player stochastic games, namely, stochastic games with independent chains and unknown transition matrices. In this class of games, players control their own internal Markov chains whose transitions do not depend on the states/actions of other players. However, players' decisions are coupled through their payoff functions. We assume players can receive only realizations of their payoffs, and that the players can not observe the states and actions of other players, nor do they know the transition probability matrices of their own Markov chain. Relying on a compact dual formulation of the game based on occupancy measures and the technique of \emph{confidence set} to maintain high-probability estimates of the unknown transition matrices, we propose a fully decentralized mirror descent algorithm to learn an $\epsilon$-NE for this class of games. The proposed algorithm has the desired properties of \emph{independence}, \emph{scalability}, and \emph{convergence}. Specifically, under no assumptions on the reward functions, we show the proposed algorithm converges in polynomial time in a weaker distance (namely, the averaged Nikaido-Isoda gap) to the set of $\epsilon$-NE policies with arbitrarily high probability. Moreover, assuming the existence of a \emph{variationally stable} Nash equilibrium policy, we show that the proposed algorithm converges asymptotically to the stable $\epsilon$-NE policy with arbitrarily high probability. In addition to Markov potential games and linear-quadratic stochastic games, this work provides another subclass of $n$-player stochastic games that, under some mild assumptions, admit polynomial-time learning algorithms for finding their stationary $\epsilon$-NE policies.  
\end{abstract}

\section{Introduction}\label{sec:intro}

Learning of equilibrium points in noncooperative games is a fundamental problem that has emerged in many disciplines, such as control and game theory, operations research, and computer science \cite{zhang2021multi,cesa2006prediction,daskalakis2021independent}. Broadly speaking, in a strategic multiagent decision-making system, one immediate goal for the agents (players) is to adaptively make decisions in order to optimize their own payoff functions. However, since the players' payoff functions may be misaligned, their decisions may not result in a Nash equilibrium (NE) -- a stable outcome in which each player's decision optimizes its own payoff given the fixed decisions of the others. Therefore, a major question is whether the players can intelligently learn how to update their strategies through interactions so that their collective decisions converge to a NE. Moreover, the convergence time to a NE should be short and scale polynomially in terms of the game parameters. In addition, since players are often selfish and interact in a fully competitive environment, their decisions must be made independently and without coordination. In other words, each player must make decisions only by observing its realized payoff and without knowing others' decisions/payoffs.\footnote{This information setting is often referred to as ``bandit information feedback" in the game-theoretic literature.} Therefore, \emph{convergence}, \emph{scalability}, and \emph{independence} are three major challenges for efficient and independent learning of NE points in large-scale noncooperative games.  

Typically, efficient and independent learning of NE is challenging, and it is known that computing NE is PPAD-hard \cite{daskalakis2009complexity} for general-sum games. The learning task is even more complex for stochastic dynamic games \cite{shapley1953stochastic,bacsar1998dynamic} where the existence of state dynamics introduces additional nonstationarity to the environment. While it is hopeless to provide a general learning scheme for all noncooperative games, it has been shown in the past literature that efficient and independent learning of NE points is still possible for special structured stochastic games such as two-player zero-sum stochastic games \cite{zhao2021provably,qiu2021provably}, mean-field and aggregative stochastic games \cite{uz2020reinforcement,meigs2019learning}, and $n$-player Markov potential games \cite{zhang2021gradient,leonardos2021global}. Expanding upon the existing literature, in this work, we study a subclass of noncooperative stochastic games, namely, stochastic games with independent chains and unknown transition matrices \cite{altman2007constrained,etesami2022learning}, and our goal is to provide scalable and independent learning algorithms for NE points. In this class of games, a set of $n$ players, each with its own finite state and action space, controls its own Markov chain, whose transition does not depend on the states/actions of other players. However, the players are coupled through their payoff functions, which depend on the states and actions of all players. We also assume that the players can not observe the states and actions of other players, nor do they know the transition probability matrices of their own Markov chain. There are many interesting real-world problems that fit into this subclass of stochastic games. The following are only two motivating examples.

\smallskip
\noindent
{\bf I) Energy Management in Smart Grids:} Consider an energy market with one utility company and a set $[n]=\{1,\ldots,n\}$ of users (players), which can both produce and consume energy (see Figure \ref{Fig:energy}). Each player generates energy using its solar panel or wind turbine and is equipped with a storage device that can store the remaining energy at the end of each day $t\in \mathbb{Z}_+$. Let $s^t_i$ denote the (quantized) amount of stored energy of player $i$ at the beginning of day $t$ with maximum storage capacity $C$. Moreover, let $g_i^t$ be a random variable denoting the amount of harvested energy for player $i$ at the end of day $t$, whose distribution is determined by the weather conditions on that day. Now if we use $a_i^t$ to denote the total amount of energy consumed by player $i$ during day $t$, then the stored energy at the end of day $t$ (or the beginning of day $t+1$) is given by $s_i^{t+1}=\min\{C, g_i^t+(s_i^t-a_i^t)^+\}$, where $(s_i^t-a_i^t)^+=\max\{0, s_i^t-a_i^t\}$. In particular, player $i$ needs to purchase $(a_i^t-s_i^t)^+$ units of energy from the utility company on day $t$ to satisfy its demand on that day. On the other hand, the utility company sets the energy price as a function of total demands $\{(a_i^t-s_i^t)^+, i\in [n]\}$, which is given by $p(a^t,s^t)$. If $u_{i}(a^t_i)$ denotes the utility that player $i$ derives by consuming $a_i^t$ units of energy, then the reward of player $i$ at time $t$ is given by $r_i(a^t,s^t)=u_{i}(a^t_i)-p(a^t,s^t)\times(a_i^t-s_i^t)^+$. In particular, if players are at distant locations, they likely experience independent weather conditions, so their transition probability models that are governed by stochasticity of $\{g^t_i, i\in [n]\}$ will be independent. In this game, players want to adopt consumption policies to maximize their aggregate rewards despite not being able to observe others' states/actions. Therefore, if the players can learn their NE policies quickly, the whole energy market will stabilize, eventually benefiting both the utility company and the players \cite{etesami2018stochastic,etesami2022learning}. 

\begin{figure}[!tbp]
	\vspace{-0.4cm}
	\centering
	\begin{minipage}[t]{0.6\textwidth}
		\includegraphics[width=\textwidth]{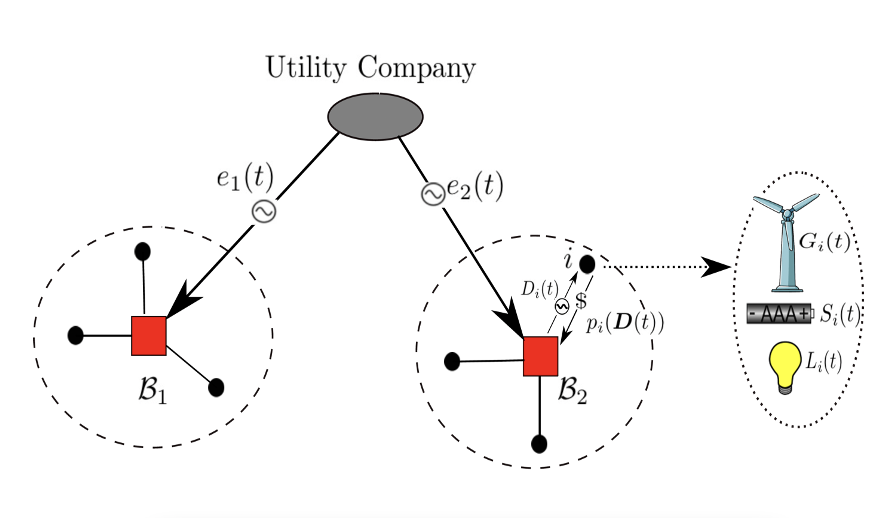}
		\caption{\footnotesize{Energy management in smart grids: Users interact with the utility company and each other to learn their equilibrium energy consumption/production levels.}}\label{Fig:energy}
	\end{minipage}
	\hfill
	\begin{minipage}[t]{0.35\textwidth}
		\includegraphics[width=\textwidth]{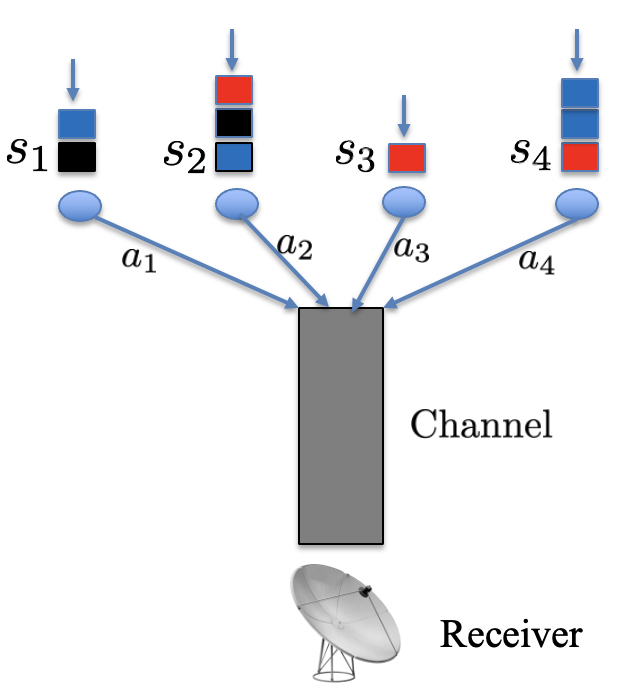}
		\caption{\footnotesize{Multi-agent wireless communication: Senders adjust their transmission powers (action) by observing their payoffs/local queues (state).}}\label{Fig:wireless}
	\end{minipage}
\end{figure}

\smallskip
\noindent
{\bf II) Power/Bandwidth Allocation in Multi-agent Wireless Networks:} As another application of the stochastic games we study in this work one can consider the design of multi-agent wireless networking algorithms for power or bandit allocation  \cite{altman2007constrained,narayanan2017large,altman2008constrained}. More precisely, consider a set $[n]$ of players (users) that share the same wireless communication channel (see Figure \ref{Fig:wireless}). Each user $i$ has a queue length $s_i$ (state) of arriving packets that are generated from a local source, which is independent of users' generating sources. Each user $i$ observes its queue length and chooses its transmission power $a_i$ (action) to send a message through the channel. The reward of each player $r_i(\cdot)$ (e.g., the probability of successful transmission) depends on the transmission power of all other users, their queue lengths, and the shared channel environment (e.g., temperature, white noise, capacity). This setting can again be formulated in the form of the stochastic game described above, in which the users want to learn their power transmission NE policies despite not knowing the transition probabilities of their local generating sources, while they only get to observe the reward of successful transmissions.

Considering the above and many other similar examples, such as multi-agent robotic navigation \cite{zhang2021multi}, our objective in this work is to extend the existing literature on learning in games and develop scalable and independent learning algorithms for obtaining NE in a subclass of stochastic games with independent chains and unknown transition matrices. Specific goals of this work are i) to design decentralized and uncoordinated strategies for learning NE points (independency); ii) to analyze the convergence of the devised learning algorithms under limited information feedback environments (convergence); and iii) to devise polynomial-time learning algorithms for computing or approximating NE strategies (scalability).

\subsection{Background and Related Work}
Since the early work on the existence of a mixed-strategy Nash equilibrium in static noncooperative
games \cite{nash1950equilibrium}, and its extension to the existence of stationary NE policies in dynamic stochastic games \cite{shapley1953stochastic}, substantial research has been done to develop scalable algorithms for computing NE points in static and dynamic environments \cite{cesa2006prediction}. In general, computing NE is PPAD-hard \cite{daskalakis2009complexity}, and it is unlikely to admit a polynomial-time algorithm. To overcome this fundamental barrier, one main approach that has been adopted in the past literature is to search for NE points in special structured games, such as \emph{potential games} or \emph{concave games}. However, such results have mainly been developed for static games, in which players repeatedly play the same game and gradually learn the underlying stationary environment. Unfortunately, an extension of such results to dynamic stochastic games \cite{shapley1953stochastic,basar1999dynamic}, in which the state of the game evolves as a result of players' past decisions and the realizations of a stochastic nature, imposes major challenges.

For dynamic stochastic games, the prior work has largely focused on the special case of two-player zero-sum stochastic games \cite{zhao2021provably,qiu2021provably,zhang2021multi,tian2021online,sayin2021decentralized,sayin2022fictitious,wei2021last}. For instance, \cite{daskalakis2021independent} provided a finite-sample NE convergence result for independent policy gradient methods in two-player zero-sum stochastic games without coordination. While two-player zero-sum stochastic games constitute an important basic setting, there are many problems with a large number of players, a situation that hinders the applicability of the existing algorithms for computing a stationary NE. To address this issue, researchers have recently developed learning algorithms for finding NE in special structured stochastic games, e.g., mean-field and aggregative stochastic games\cite{zhang2021multi,uz2020reinforcement,meigs2019learning}. The main underlying assumption in mean-field games and aggregative games is that the individual actions of the players do not play a major role in the evolution of the state dynamics, but rather that the mean/aggregate of their actions is the driving force of the dynamics. Such an approach may simplify the learning task by allowing the players to focus on learning the mean-field trajectory of the actions/states rather than individual actions/states.. \cite{hambly2021policy} shows that policy gradient methods can find a NE in $n$-player general-sum linear-quadratic games, and an application of reinforcement learning for finding NE in linear-quadratic mean-field games has been studied in \cite{uz2020reinforcement}. Moreover, \cite{zhang2021gradient,leonardos2021global} show that $n$-player Markov potential games, an extension of static potential games to dynamic stochastic games, admit polynomial-time algorithms for computing their NE policies. Unfortunately, the class of Markov potential games is very restrictive because it requires strong assumptions on the existence of a general potential function. In fact, even establishing the existence of such a potential function could be a challenging task \cite{macua2018learning,mguni2021learning}. 

This work is also closely related to the large body of literature on reinforcement learning (RL), especially works on RL for adversarial Markov Decision Processes (MDPs) \cite{zhang2021multi,lee2020bias}. MDPs are general frameworks that can model many real-world decision-making problems in the face of uncertainties \cite{chen2021primal,zhang2021multi,jin2020efficiently}, and these works have raised substantial interest in developing efficient learning algorithms for computing optimal stationary policies in single-agent MDPs \cite{agarwal2021theory,wang2017primal,cardoso2019large}. In the problem of learning adversarial MDPs, where the losses can change arbitrarily between episodes, the goal of a learning algorithm is typically to achieve a small \emph{regret} compared to the best fixed policy in hindsight. In the work, we borrow techniques such as \emph{occupancy measure} and \emph{confidence set} \cite{jin2020learning,rosenberg2019online} from the literature on learning adversarial MDPs.

There has been a line of prior research on the study of decentralized stochastic games with independent chains \cite{altman2007constrained,singh2014characterization,qiu2021provably,zhang2022constrained,etesami2022learning}. Specifically, \cite{altman2007constrained} showed the existence of a NE for the class of stochastic games with independent chains.\footnote{They called it constrained cost-coupled stochastic games with independent state processes with some additional constraints.} The work \cite{singh2014characterization} showed that the set of stationary NE for the class of games can be characterized via the global minimizers of a certain non-convex mathematical program. \cite{qiu2021provably} showed that for two-player zero-sum stochastic games with independent chains, assuming players' knowledge about the opponent's past strategy, fictitious play policy optimization algorithms can achieve $\tilde{\O}(\sqrt{T})$ regrets. Recently, for $n$-player decentralized stochastic games with independent chains, relying on a dual formulation of the game based on occupancy measures, \cite{etesami2022learning} proposed polynomial-time learning algorithms based on dual averaging and dual mirror descent, which converge in terms of the averaged Nikaido-Isoda distance to the set of $\epsilon$-NE policies. However, except for \cite{qiu2021provably}, which essentially deals with two-player zero-sum stochastic games and uses a different algorithm than ours based on $Q$-learning and fictitious play, all of the aforementioned works assume players' prior knowledge of the transition probability matrices of their own Markov chain, which is somewhat restrictive in practice. Moreover, there was no algorithm with an asymptotic convergence guarantee to NE policies for the class of $n$-player stochastic games with independent chains.

\subsection{Contributions and Organization}

We consider the class of decentralized stochastic games with independent chains and unknown transition matrices. Relying on a dual formulation of the original stochastic game based on occupancy measures and introducing confidence sets to maintain high-probability estimates of the unknown transition matrices, we propose a Decentralized Mirror Descent algorithm to learn an $\epsilon$-NE. The proposed algorithm has the desired properties of independence, scalability, and convergence. Our contribution can be summarized as follows:
\begin{itemize}
	\item We propose a learning algorithm that is simple, easy to implement, and works in a fully decentralized manner. The only coordination needed is a simple signaling mechanism to indicate the end of each episode among players, which can be further relaxed by allowing an extra error term in the equilibrium computation. 
	\item Under no assumptions on the reward functions, we show the proposed algorithm is a polynomial-time learning algorithm that converges in a weaker distance (namely, the averaged Nikaido-Isoda gap) to the set of $\epsilon$-NE policies with probability at least $1-\gamma$ after at most the following number of episodes 
	$$K = \tilde{\O}\Big(\max\Big\{\frac{\big(\sum_{i=1}^n|S_i|\sqrt{\ln\frac{n|A_i||S_i|^2}{\gamma}}\big)^2}{(1-e^{-1/\tau})^2\epsilon^2},\frac{\big(\sum_{i=1}^n|A_i||S_i|^2\big)^2}{\epsilon^2}\Big\}\Big),$$
where the length of each episode is also bounded polynomially in terms of the game parameters. 
	\item Under an extra assumption on the payoffs, namely, the existence of a globally stable NE policy, which is a relaxation of the well-known monotonicity condition, we show that the proposed algorithm converges asymptotically to an $\epsilon$-NE with arbitrarily high probability.
\end{itemize}

The rest of the paper is organized as follows.
In Section \ref{sec:formulation}, we introduce the class of stochastic games with independent chains and unknown transition matrices. In Section \ref{sec:dual}, we provide a dual formulation for such games and establish several preliminary results. In Section \ref{sec:alg-design}, we develop a decentralized online mirror descent algorithm for learning $\epsilon$-NE policies. In Section \ref{sec:finite_analysis}, we present a finite time convergence analysis of the proposed algorithm by showing a polynomial-time convergence rate in terms of the averaged Nikaido-Isoda gap function to the set of $\epsilon$-NE policies. This is without any assumptions about the reward functions. In Section \ref{sec:asymptotic}, we show the proposed algorithm converges asymptotically to an $\epsilon$-NE with arbitrarily high probability under the extra variational stability assumption. We conclude the paper by identifying future research directions in Section \ref{sec:conclusion}. Omitted proofs and auxiliary lemmas can be found in Appendix A and Appendix B.

\section{Problem Formulation}\label{sec:formulation}

We consider an $n$-player stochastic game with independent and unknown state transitions, which is described by the tuple $(S_i,A_i,r_i,P_i)_{i=1}^n$, as follows.
\begin{itemize}
	\item $S_i$ is the finite set of states for player $i$ with elements $s_i\in S_i$. We denote the joint state set of all the players by $S = \prod_{i = 1}^n S_i$ with elements $\bs{s}\in S$, where $\bs{s}= (s_1,\ldots,s_n)$. 
	\item $A_i$ is the finite set of actions for player $i$ with elements $a_i\in A_i$. We denote the joint action set of all the players by $A = \prod_{i = 1}^n A_i$, and the elements of $A$ are denoted by $\bs{a} = (a_1,\ldots,a_n)$. 
	\item $r_i : S\times A \to [0, 1]$ is the reward function for player $i$, where $r_i(\bs{s},\bs{a})$ is the immediate reward received by player $i$ when the states of the players are $\bs{s} = (s_1,\ldots,s_n)$, and the actions taken by them are given by the action profile $\bs{a} = (a_1,\ldots,a_n)$.
	\item $P_i$ is the transition probability matrix for player $i$, where $P_i(s_i'|s_i,a_i)$ is the probability that the state of player $i$ moves from $s_i$ to $s_i'$ if she chooses action $a_i$. Crucial to this work, we assume that $P_i$ is \emph{unknown} to player $i$ and is \emph{independent} of other players' transition probability matrices. 
\end{itemize}

\begin{assumption}\label{ass:indep}
We assume that the joint transition probability matrix $P(s'|s,a)$ can be factored into independent components $P(s'|s,a)=\prod_{i=1}^n P_i(s'_i|s_i,a_i)$, where $P_i(s'_i|s_i,a_i)$ is the transition probability matrix for player $i$.
\end{assumption}

At any time $t$, the information available to player $i$ is given by the history of its realized states, actions, and rewards, i.e., $\H_i^t = \{s^l_i,a^l_i,r_i(s^l,a^l) : l = 0,1,\ldots ,t-1\}\cup\{s^t_i\}$.
Given  information set $\H_i^t$, player $i$ takes an action $a^t_i$ based on her current policy $\pi^t_i(\cdot | \H_i^t)$ which is a probability measure over $A_i$ and receives a reward $r_i(s^t,a^t)$, which also depends on other players’ states and actions. After that, the state of player $i$ changes from $s^t_i$ to a new state $s^{t+1}_i$ with probability $P_i(s^{t+1}_i|s^t_i,a^t_i)$. 

A sequence of probability measures $\pi_i = \{\pi_i^t, t = 0, 1, \ldots\}$ over $A_i$ that, at each
time $t$ selects an action $a_i \in A_i $ based on past observations $\H_i^t$ with probability $\pi^t_i(\cdot | \H_i^t)$, consists a general policy for player $i$. However, use of general policies is often computationally expensive, and in practical applications, players are interested in the easily implementable \emph{stationary} policies, as defined next.

\begin{definition}
	A policy $\pi_i$ for player $i$ is called stationary if the probability $\pi_i^t(a_i|\H^t)$ of choosing action $a_i$ at time $t$ depends only on the current state $s^t_i = s_i$, and is independent of the time $t$. In the case of the stationary policy, we use $\pi_i(a_i|s_i)$ to denote this time-independent probability.
\end{definition}


Given some initial state $s^0$, the objective for each player $i\in [n]$ is to choose a stationary policy $\pi_i$ that maximizes its long-term expected average payoff given by
\begin{align}\label{eq:aggregate-reward}
V_i(\pi_i,\pi_{-i})=\mathbb{E}\Big[\lim_{T\to \infty}\frac{1}{T}\sum_{t=0}^{T}r_i(s^t,a^t)\Big],
\end{align} 
where $\pi_{-i}=(\pi_j, j\neq i)$,\footnote{More generally, given a vector $v$, we let $v_{-i}=(v_j, j\neq i)$ be the vector of all coordinates in $v$ other than the $i$th one.} and the expectation is with respect to the randomness introduced by players' internal chains $(P_1,\ldots,P_n)$ and their policies $\bs{\pi}=(\pi_1,\ldots,\pi_n)$. 

Next, in order to be able to establish meaningful convergence/learning results, we impose the following assumption througout this work. 

\begin{assumption}\label{ass:ergodic}
	For any player $i$ and any stationary policy $\pi_i$ chosen by that player, the induced Markov chain with transition probabilities $P^{\pi_i} (s'_i|s_i) = \sum_{a_i\in A_i}P_i(s'_i|a_i,s_i)\pi_i(a_i|s_i)$, is ergodic, and its mixing time is uniformly bounded above by some parameter $\tau$; that is, 
	\begin{align*}
		\|(\nu-\nu')P^{\pi_i}\|_1\le e^{-1/\tau}\|\nu-\nu'\|_1,\qquad \forall i,\pi_i,\nu,\nu'\in\Delta(S_i).
	\end{align*}
\end{assumption}

In fact, Assumption \ref{ass:ergodic} is a standard assumption used in the MDP literature and is much needed. Otherwise, if the transition probability matrix $P_i$ of a player $i$ is such that for some policy $\pi_i$ the induced chain $P^{\pi_i}$ takes an arbitrarily large time to mix, then there is no hope that player $i$ can evaluate the performance of policy $\pi_i$ in a reasonably short time. As is shown in the next section, under the ergodicity Assumption \ref{ass:ergodic}, for any stationary policy profile $\bs{\pi}$, the limit in \eqref{eq:aggregate-reward} indeed exists and equals \eqref{eq:linear-occupation-form1}. This fully characterizes an $n$-player stochastic game with initial state $s^0$, in which each player $i$ wants to choose a stationary policy $\pi_i$ to maximize its expected aggregate payoff $V_i(\pi_i,\pi_{-i})$. In the remainder of the paper, we shall refer to the above payoff-coupled stochastic game with independent chains and unknown transitions as $\mathcal{G}=([n], \pi,\{V_i(\pi)\}_{i\in [n]})$.

\begin{definition}
For a policy profile $\bs{\pi}^*=(\pi^*_1,\ldots,\pi^*_n)$, $\pi^*_i$ is called a best response policy of $\pi_{-i}^*$ if $V_i(\pi^*_i,\pi^*_{-i})\ge V_i(\pi_i,\pi^*_{-i})$ for any policy $\pi_i$. It is called an $\epsilon$-best response policy if $V_i(\pi^*_i,\pi^*_{-i})\ge V_i(\pi_i,\pi^*_{-i})-\epsilon$  for any policy $\pi_i$. The policy profile $\bs{\pi}^*=(\pi^*_1,\ldots,\pi^*_n)$ is called a Nash equilibrium (NE) for the game $\mathcal{G}$ if for any $i$, $\pi^*_i$ is a best response policy of $\pi_{-i}^*$. It is called an $\epsilon$-NE if for any $i$, $\pi^*_i$ is an $\epsilon$-best response policy of $\pi_{-i}^*$.
\end{definition}

The main objective of this work is to develop a decentralized and salable learning algorithm such that, if followed by the players independently, it brings the system to an $\epsilon$-NE stationary policy.



\bigskip
\section{A Dual Formulation and Preliminaries}\label{sec:dual}

In this section, we provide an alternative dual formulation for the stochastic game $\G$ based on \emph{occupancy measures}\cite{altman2021constrained}. Intuitively, from player $j$'s point of view, its long-term expected average payoff depends on the proportion of time that player $j$ spends in state $s_j$ and takes action $a_j$, denoted by its occupancy measure. Thus, the policy optimization for player $j$ can be viewed as an optimization problem in the space of occupancy measures, where players want to force their chains to spend most of their time in high-reward states. An advantage of optimization in terms of occupancy measures is that due to the independence of players' internal chains, the payoff functions admit a simple decomposable form, which is easier to analyze than the original policy variables. Moreover, the simple structure of the payoff functions in the space of occupancy measures allows us to tackle the learning problem using rich literature from online learning. We shall use this dual formulation to develop learning algorithms for finding a stationary $\epsilon$-NE.

\subsection{Occupancy Measure}

For a given MDP with a transition probability matrix $P$ and any stationary policy $\pi$, one can associate with $P$ and $\pi$ three notions of occupancy measures $\nu :  S \to [0, 1]$, $\rho :  S\times A\to [0, 1]$, and $q: S\times A\times S \to [0, 1]$ , as\footnote{The first type of occupancy measure $\nu$ is often referred to as the stationary distribution in the past literature.}
\begin{align}
	&\nu(s) = \lim_{t\to \infty}\frac{1}{T}\sum_{t=0}^{T}\mathbb{P}(s^t=s) \label{eq:niu}, \\ 
	&\rho(s,a) = \lim_{t\to \infty}\frac{1}{T}\sum_{t=0}^{T}\mathbb{P}(s^t=s,a^t=a) \label{eq:rho}, \\ 
	&q(s,a,s') = \lim_{t\to \infty}\frac{1}{T}\sum_{t=0}^{T}\mathbb{P}(s^t=s,a^t=a,s^{t+1} = s'). \label{eq:q}
\end{align}
Intuitively, $\nu(s), \rho(s,a)$ and $q(s,a,s')$ are the long-term average proportion of time of encountering the state $s$, state-action pair $(s,a)$, and state-action-next-state triple $(s,a,s')$, when executing policy $\pi$ in an MDP with transition probability matrix $P$. It can be readily shown that under Assumption \ref{ass:ergodic}, the limits in \eqref{eq:niu}, \eqref{eq:rho}, and \eqref{eq:q} indeed exist, and moreover, the following relations always hold:
\begin{align*}
	\rho(s,a) = \nu(s) \pi(a|s),\qquad q(s,a,s') = \rho(s,a) P(s'|s,a).
\end{align*}

In this work, we are primarily concerned with the occupancy measure $q(s,a,s')$ due to players not knowing their independent transition probability matrix $P_i$ in the stochastic game $\G$. In the following, we provide conditions that fully characterize the set of feasible occupancy measures $q$.

\begin{definition}\label{def:Delta}
We define the polytope of feasible occupancy measures, denoted by $\Delta$, as 
\begin{align}\label{eq:Delta}
	\Delta = \Big\{q\in [0, 1]^{|S \times A \times S|} \quad : \sum_{s,a,s'}q(s,a,s') = 1,\quad\sum_{s',a}q(s',a,s) = \sum_{a,s'}q(s,a,s'),\quad \forall s\in S\Big\}. 
\end{align}
For any $q \in \Delta$, we define its induced transition probability matrix $P^q$ and its induced stationary policy $\pi^q$ by
\begin{align*}
P^q(s'|s,a) = \frac{q(s,a,s')}{\sum_{s'}q(s,a,s')} \ \forall s,a,s', \ \ \ \ \ \ \pi^q(a|s)=\frac{\sum_{s'}q(s,a,s')}{\sum_{a',s'}q(s,a',s')}\ \forall s,a.
\end{align*}
Moreover, for a fixed transition probability matrix $P$, we denote by $\Delta(P ) \subset \Delta$ the set of occupancy measures whose induced transition probability matrix $P^q$ is exactly $P$. Similarly, we denote by $\Delta(\CP) \subset \Delta$ the set of occupancy measures whose induced transition probability matrix $P^q$ belongs to a set of transition matrices $\CP$.
\end{definition}

The reason why \eqref{eq:Delta} provides the set of feasible occupancy measures is because any feasible occupancy measure $q$ should be a valid probability distribution, i.e., $\sum_{s,a,s'}q(s,a,s') = 1$, and moreover, for any state $s\in S$,  the probability of entering it should equal to the probability of leaving it, i.e., $\sum_{s',a}q(s',a,s) = \sum_{a,s'}q(s,a,s')\ \forall s\in S.$ Given the above definition, we have the following useful lemma from \cite{altman2021constrained} .

\begin{lemma}\label{lem:feasible}
	If a function $q:  S\times A\times S \to [0, 1]$ belongs to the feasible occupancy polytope \eqref{eq:Delta}, then it is exactly the occupancy measure associated with its induced transition probability matrix $P^q$ and stationary policy $\pi^q$. Specifically, we have \eqref{eq:q} hold if one executes policy $\pi^q$ in an MDP with transition probability matrix $P^q$.
\end{lemma}

\subsection{A Dual Formulation}
We now give a dual formulation of the stochastic game $\G$ based on occupancy measures. It is shown in \cite{etesami2022learning} that under the ergodicity Assumption \ref{ass:ergodic}, due to the independency of players’ internal chains, the payoff functions admit a simple decomposable form in terms of occupancy measures $\rho$. However, since we are interested in stochastic games with unknown transition probabilities, we first extend this result in terms of occupancy measures $q$. Specifically, assume that each player $i$ is following a stationary policy $\pi_i$, and let $\nu_i$, $\rho_i$, and $q_i$ be the corresponding occupancy measures given in \eqref{eq:niu}, \eqref{eq:rho}, and \eqref{eq:q}, that are induced by $P_i$ and following the stationary policy $\pi_i$. Then, we can write the expected average payoff of player $i$, $V_i(\pi_i,\pi_{-i})$ equivalently as a function of $\rho_i$ and $\rho_{-i}$.
\begin{proposition}\label{prop:dual_formulation}
	Let Assumptions \ref{ass:indep}, \ref{ass:ergodic} hold, and assume that each player $i$ follows a stationary policy $\pi_i$. Then,
	\begin{align}\label{eq:linear-occupation-form2}
	V_i(\pi_i,\pi_{-i}) = V_i(q_i,q_{-i})\triangleq\sum_{s,a}\prod_{j=1}^n\sum_{s'_j}q_j(s_j,a_j,s'_j)r_i(s,a)=\langle q_i, l_i(q_{-i}) \rangle,
\end{align} 
where $l_i(q_{-i})$ is defined to be a vector of dimension $|A_i||S_i|^2$ whose $(s_i,a_i,s'_i)$-th coordinate is given by
\begin{align}\label{eq:l_def}
	l_i(q_{-i})_{(s_i,a_i,s'_i)}=\sum_{s_{-i},a_{-i}}\prod_{j\neq i}\sum_{s'_j}q_j(s_j,a_j,s'_j) r_i(s,a).
\end{align}
Moreover, we have $l_i(q_{-i})_{(s_i,a_i,s'_i)}=\sum_{s_{-i},a_{-i}}\prod_{j\neq i}\rho_j(s_j,a_j) r_i(s,a)\ \forall s_i,a_i,s'_i.$
\end{proposition}
\begin{proof}
It has been shown in \cite{etesami2022learning} that under Assumptions \ref{ass:indep} and \ref{ass:ergodic}, if each player $i$ follows a stationary policy $\pi_i$, then
	\begin{align}\label{eq:linear-occupation-form1}
		V_i(\pi_i,\pi_{-i}) = V_i(\rho_i,\rho_{-i})=\sum_{s,a}\prod_{j=1}^n\rho_j(s_j,a_j)r_i(s,a)=\langle \rho_i, v_i(\rho_{-i}) \rangle,
	\end{align} 
	where $v_i(\rho_{-i})$ is defined to be a vector of dimension $|S_i||A_i|$ whose $(s_i,a_i)$-th coordinate is given by
	\begin{align}\label{eq:v-rho-coordinate}
		v_i(\rho_{-i})_{(s_i,a_i)}=\sum_{s_{-i},a_{-i}}\prod_{j\neq i}\rho_j(s_j,a_j) r_i(s,a).
	\end{align}
 Since we have $\rho_i(s_i, a_i) = \sum_{s'_i}q_i(s_i, a_i,s'_i)$, by substituting this relation into \eqref{eq:linear-occupation-form1} we obtain the equivalent forms  \eqref{eq:linear-occupation-form2} and \eqref{eq:l_def} in terms of the $q$ variables. Finally, from \eqref{eq:l_def} and \eqref{eq:v-rho-coordinate} and the fact that $\rho_i(s_i, a_i) = \sum_{s'_i}q_i(s_i, a_i,s'_i)$, it is easy to see that $l_i(q_{-i})_{(s_i,a_i,s'_i)}=v_i(\rho_{-i})_{(s_i,a_i)},\ \forall s_i,a_i,s'_i.$
\end{proof}


Using Lemma \ref{lem:feasible} and Proposition \ref{prop:dual_formulation}, the problem of finding the optimal stationary policies for the players reduces to one of finding the optimal feasible occupancy measures for them. In fact, with the dual formulation in hand, we can formulate an equivalent virtual game associated with the stochastic game $\G$ as follows:

\smallskip
\begin{definition}\label{def:v_game}
Let 
$$\Delta_i\!=\! \big\{q_i\!\in\! [0, 1]^{|A_i||S_i|^2}\!\!: \sum_{s_i,a_i,s'_i}q_i(s_i,a_i,s'_i)\!=\! 1, \sum_{s'_i,a_i}\!q_i(s'_i,a_i,s_i)\!=\! \sum_{s'_i,a_i}\!q_i(s_i,a_i,s'_i) ,\forall s_i\big\}$$
be the feasible occupancy polytope for player $i$. Moreover, denote by $\Delta_i(P_i) \subset \Delta_i$ the set of feasible occupancy measures whose induced transition probability matrix $P^{q_i}$ is exactly $P_i$. The virtual game $\V=([n], {q},\{V_i({q})\}_{i\in [n]})$ associated with the stochastic game $\G $ is an $n$-player continuous-action static game, where the action of player $i$ is to choose an $q_i$ from its action set $\Delta_i (P_i)$, and its payoff function is given by \eqref{eq:linear-occupation-form2}.
\end{definition} 

\smallskip
Ideally, we would like every player to work with the virtual game $\V$ with action set $\Delta_i(P_i)$ as it admits a payoff function that is linear with respect to the player's action, hence making it amenable to the use of online learning algorithms. However, this can not be performed as $P_i$ is not known to player $i$ so the player can not compute $\Delta_i(P_i)$. Nevertheless, observe that once each player $i$ has decided on her occupancy measure $\hat{q}_i \in \Delta_i$ (which may not belong to $\Delta_i(P_i)$), then the game $\G$ is fully determined by the players' policies $\{\pi^{\hat{q}_i}_i\}_{i=1}^n$, where $\pi^{\hat{q}_i}_i$ is the stationary policy induced by $\hat{q}_i$. In this regard, with some abuse of notations, the payoff function of player $i$ is given by $$V_i(\hat{q}_i,\hat{q}_{-i}) = V_i(\pi_i^{\hat{q}_i}, \pi_{-i}^{\hat{q}_{-i}}),$$ where $\hat{q}_i \in \Delta_i$, and $V_i(\pi_i^{\hat{q}_i}, \pi_{-i}^{\hat{q}_{-i}})$ is as defined in \eqref{eq:aggregate-reward}. When $\hat{q}_i \in \Delta_i(P_i)$, we also have $V_i(\hat{q}_i,\hat{q}_{-i}) = V_i(\pi_i^{\hat{q}_i}, \pi_{-i}^{\hat{q}_{-i}}) = \langle {\hq}_i, l_i({\hq}_{-i}) \rangle $ as defined in \eqref{eq:linear-occupation-form2}. In the subsequent sections, we will develop a learning algorithm such that each player $i$ updates her occupancy measure $\hat{q}_i$ iteratively to maximize $V_i(\hat{q}_i,\hat{q}_{-i})$, while improving its estimate about the true underlying transition probability matrix $P_i$.

\section{A Learning Algorithm for $\epsilon$-NE Policies}\label{sec:alg-design}

In this section, we develop our learning algorithm for the stochastic game $\mathcal{G}$. The algorithm proceeds in different episodes, each containing a random number of time instances. The main idea is that each player $i$ will use \emph{confidence sets} and \emph{online mirror descent} (OMD) to learn an occupancy measure $\hat{q}_i$ such that (i) its induced transition probability matrix $P^{\hat{q}_i}_i$ approximates the true transition probability matrix $P_i$, and (ii) its induced stationary policy $\pi^{\hat{q}_i}_i$ approximates player $i$'s best response to $\pi_{-i}^{\hat{q}_{-i}}$. The complete pseudo-code of the proposed learning algorithm is presented in Algorithm \ref{alg:main}. To describe the learning algorithm in detail, we first consider the following definition of a \emph{shrunk polytope}.

\begin{definition}[Shrunk Polytope]
    \label{def:shrunk}
	Given $0<\delta_i < 1$, we define $\Delta_{i,\delta_i} \triangleq \{q_i\in \Delta_i  :  \sum_{s'_i}q_i(s_i, a_i,s'_i)\ge \delta_i, \forall s_i,a_i\}$
	to be the shrunk polytope of feasible occupancy measures for player $i$. Equivelently, we can write $\Delta_{i,\delta_i} = \Delta_i \cap \{q_i:\rho_i \ge \delta_i\mathbf{1}\}$, where $\mathbf{1}$ is the column vector of all ones of dimension $|S_i||A_i|$ and $\rho_i(s_i, a_i) = \sum_{s'_i}q_i(s_i, a_i,s'_i)$. Moreover, for a fixed transition probability matrix $P_i$ or a set of transition probability matrices $\CP_i$, we define $\Delta_{i,\delta_i}(P_i) \subseteq \Delta_{i,\delta_i}$ or $\Delta_{i,\delta_i}(\CP_i)\subseteq \Delta_{i,\delta_i}$ as the set of occupancy measures $q_i$ whose induced transition probability matrix $P^{q_i}$ equals $P_i$ or belongs to the set $\CP_i$, respectively.
\end{definition}

Restricting player $i$'s occupancy measures to be in $\Delta_{i,\delta_i}$ ensures that player $i$ uses stationary policies that choose any action with probability at least $\delta_i$, hence encouraging exploration during the learning process. Thanks to the continuity of the payoff functions, working with shrunk polytope $\Delta_{i,\delta_i}$ with a sufficiently small threshold $\delta_i$ can only result in a negligible loss in players' payoffs, as shown in the following lemma \cite[Lemma 2]{etesami2022learning}.

\begin{lemma}\label{lem:shrunk}
For any $\epsilon > 0$, there exist polynomial-time computable thresholds $\{\delta_i > 0, i \in [n]\}$, such that
	\begin{align}
		\label{eq:shrunk}
		\max_{q_i'\in\Delta_{i,\delta_i}(P_i)}V_i(q_i',{q}_{-i})\ge \max_{q_i'\in\Delta_{i}(P_i)}V_i(q_i',{q}_{-i})-\epsilon, \ \ \forall q\in \Delta_1\times\cdots\times\Delta_n. 
	\end{align}
\end{lemma}

Finally, we consider the following ``nondegeneracy" assumption on players' internal chains, which requires that with some positive probability $\alpha_i>0$, all states are reachable for each player $i$ and under all policies. Assumption \ref{ass:minprob} serves to provide an upper bound for the length of each episode in our learning algorithm and can be viewed as a relaxation of that made in other works for the case of single-agent MDPs \cite{rosenberg2019online,neu2010online}. 

\begin{assumption}
	\label{ass:minprob}
	There exists some $\alpha >0$ such that for every player $i$,   $\sum_{a_i}P_i(s_i'|s_i,a_i)>\alpha ,\ \forall s_i,s_i'$.
\end{assumption}

Now, we are ready to describe our main distributed learning algorithm. Each player $i$ performs two tasks in parallel: (i) maintains and updates a confidence set $\CP_i$ of its own (unknown) transition probability matrix $P_i$, and (ii) uses an OMD rule to update the occupancy measure $\hat{q}_i\in \Delta(\CP_i)$. In the following, we describe each of these tasks in detail.

\subsection{Confidence Set}
 For each player $i$, the algorithm maintains counters $N_i(s_i,a_i)$ and $M_i(s_i,a_i,s'_i)$  to record the total number of visits of each state-action pair $(s_i, a_i)$ and each state-action-state triple $(s_i,a_i,s'_i)$ so far, respectively. A \emph{confidence set} $\CP_i$, which includes all transition probability matrices that are close to $P_i$ with high confidence, is maintained and updated for each episode. Specifically,  at the end of each episode $k\ge1$, player $i$ will compute the empirical transition probability matrix $\bar{P}_i^k(s'_i|s_i,a_i) = \frac{M_i^k(s_i,a_i,s'_i)}{\max\{1,N_i^k(s_i,a_i)\}}$ from the current counters $N_i^k(s_i,a_i)$ and $M_i^k(s_i,a_i,s'_i)$, and will update the confidence set for episode $k$ as 
\begin{align}
	\label{eq:confupdate}
	\mathcal{P}_i^k=\Big\{\hat{P}: |\hat{P}(s'_i|s_i,a_i)-\bar{P}_i(s'_i|s_i,a_i)|\le \epsilon_i^k(s'_i|s_i,a_i),\ \forall s'_i,s_i,a_i\Big\} \cap \CP_i^{k-1},
\end{align}
where $\epsilon_i^k(\cdot)$ is a parameter that will be determined later. Note that the confidence set in \eqref{eq:confupdate} is also a polytope with an efficient description in terms of the problem parameters.  

\subsection{Online Mirror Descent (OMD)}

The OMD component of our algorithm is similar to \cite{etesami2022learning}. Given any desired accuracy $\epsilon > 0$ for an $\epsilon$-NE, each player first uses Lemma \ref{lem:shrunk} to determine a threshold $\delta_i$, and chooses an initial occupancy measure $\hat{q}^1_i \in \Delta_{i,\delta_i}$. During each episode $k$, player $i$ takes actions according to the stationary policy $\pi_i^k:= \pi^{\hat{q}^k_i}_i$. The episode continues until each player $i$ has visited all its state-action pairs $(s_i,a_i)$ at least once. At the end of the episode, player $i$ will first update the confidence set $\mathcal{P}_i^k$ as in \eqref{eq:confupdate}, and then will update its occupancy measure $\hat{q}^{k+1}_i$ using OMD:
\begin{align*}
	\hq^{k+1}_i=\argmax_{\hq_i\in \Delta_{i,\delta_i}(\CP_i^k)}\big\{\eta^k\langle \hq_i, R^k_i\rangle -D_{h_i}(\hq_i||\hq_i^k)\big\},
\end{align*}
where $\eta^k$ is the stepsize, $D_{h_i}(p||q) \triangleq h_i(p)-h_i(q)-\la \nabla h_i(q),p-q\ra$ is the \emph{Bregman divergence} induced by a $\mu$-strongly convex regularizer $h_i(\cdot)$, and $R_i^k$ is an estimator for the gradient of the payoff function $V_i(\pi^{\hat{q}^k_i}_i,\pi^{\hat{q}^k_{-i}}_{-i})$ constructed using the collected samples of the reward $r_i$ during episode $k$. Since $\hat{q}^k_i \in \Delta_{i,\delta_i},$ from Assumption \ref{ass:minprob}, one can show that at any time $t$, $\P(s_i^t = s_i)\ge \alpha \delta_i\ \forall s_i$. As a result, the expected length of each episode $k$ is bounded above by the cover time of the underlying Markov chain induced by the policy $\pi_i^k$ and is at most $\tilde{\O}(\max_i\frac{ |S_i|}{\alpha \delta_i^2})$, where $\tilde{\O}(\cdot)$ hides logarithmic terms.

\begin{algorithm}[t!]\caption{A Decentralized Online Mirror Descent Algorithm for Player $i$}\label{alg:main}

\noindent
{\bf Input:} Initial occupancy measure $\hat{q}^1_i  = \frac{1}{|A_i||S_i|^2}\cdot \mathbf{1}$, counters $N_i(s_i,a_i)=0$, $M_i(s_i,a_i,s'_i)=0$, step-size sequence $\{\eta^k\}_{k=1}^K$, mixing time thresholds $\{d^k\}_{k=1}^K$, and a $\mu$-strongly convex regularizer $h_i:\Delta_{i,\delta_i}\to \mathbb{R}$. 

\medskip
{\bf For} $k=1,\ldots,K$, do the following:

\begin{itemize} 
\item[$\bullet$] At the start of episode $k$, compute $\pi^k_i = \pi^{\hat{q}^k_i}$, i.e.,
\begin{align}
	\label{eq:pi_i}
\pi^k_i(a_i|s_i)=\frac{\sum_{s'_i}\hq^k_i(a_i,s_i,s'_i)}{\sum_{a'_i,s'_i}\hq^k_i(s_i,a'_i,s'_i)}  \ \ \forall s_i\in S_i, a_i\in A_i,
\end{align}
and keep playing according to this stationary policy $\pi_i^k$ during expisode $k$. Update counters $N_i(s_i,a_i)$ and $M_i(s_i,a_i,s'_i)$ at each step. 
\item[$\bullet$] Let $\tau^k_i\ge d^k$ be the first (random) time such that all state-action pairs $(s_i,a_i)$ are visited during steps $[d^k, \tau_i^k]$. Episode $k$ terminates after $\tau^{k}\!=\!\max_i \tau_i^k$ steps. 

\item[$\bullet$] Let $X'_i=S_i\times A_i$, and $R^k_i\in\mathbb{R}_+^{|S_i||A_i|}$ be a random vector (initially set to zero), which is constructed sequentially during the sampling interval $[\tau^k+d,\tau^{k+1}]$ as follows: 
\begin{itemize}
\item {\bf For} $t=d^k,\ldots,\tau^{k}$ and while $X_i\neq \emptyset$, player $i$ picks an action $a^t_i$ according to $\pi_i^k(\cdot|s_i^t)$, and observes the payoff $r_i(s^t,a^t)$ and its next state $s_i^{t+1}$. If $(s_i^t,a^t_i)\in X_i$, then update $X_i=X_i\setminus \{(s_i^t,a^t_i)\}$, and let
\begin{align}\label{eq:unbiased-gradient-vector}
R^k_i=R^k_i+r_i(s^t,a^t)\ \boldsymbol{\rm e}_{(s^t_i,a^t_i)},
\end{align} 
where $\boldsymbol{\rm e}_{(s^t_i,a^t_i)}$ is the basis vector with all entries being zero except that the $(s^t_i,a^t_i)$-th entry is 1.\item[]{\bf End For}  
\end{itemize}
\item Expand $R_i^k$ from $\R^{|S_i\times A_i|}$ to $\R^{|S_i\times A_i\times S_i|}$, $i.e.$, $R_i^k(s_i,a_i,s'_i) = R_i^k(s_i,a_i), \forall s_i,a_i,s'_i$.
\item[$\bullet$] At the end of episode $k$, update the confidence set $\mathcal{P}_i^k$ as in \eqref{eq:confupdate}, and the occupancy measure using
\begin{align}\label{eq:alg-rho-update}
	\hq^{k+1}_i=\argmax_{\hq_i\in \Delta_{i,\delta_i}(\CP_i^k)}\big\{\eta^k\langle \hq_i, R^k_i\rangle -D_{h_i}(\hq_i||\hq_i^k)\big\}.
\end{align}  
\end{itemize}
{\bf End For}
\end{algorithm}

\section{Polynomial-Time Convergence Rate Using Averaged Nikaido-Isoda Gap Function}
\label{sec:finite_analysis}

In this section, we analyze the finite time convergence properties of Algorithm \ref{alg:main}. First we need to introduce some notations. Following Algorithm \ref{alg:main}, every player $i$ will hold a occupancy measure $\hq_i^k$ during episode $k$, and will play according to policy $\pi^k_i = \pi^{\hat{q}^k_i}$ as defined in \eqref{eq:pi_i}. We denote the occupancy measures induced by $\pi_i^k$ and the unknown transition probability matrix $P_i$ over the spaces $S_i$, $S_i\times A_i$, and $S_i\times A_i \times S_i$, by $\nu_i^k$, $\rho_i^k$, and $q_i^k$, respectively. Notice that $\hq_i^k$ is not necessarily the true $q_i^k$ as it is induced by $\pi_i^k$ and an estimated $\hat{P}_i^k\in \CP_i^k$. We also denote the occupancy measures induced by $\pi_i^k$ and $\hat{P}_i^k$ over the spaces $S_i$ and $S_i\times A_i$ by $\hat{\nu}_i^k$ and $\hat{\rho}_i^k$, respectively. Moreover, we let $p_i^k$ and $ \hat{p}_i^k$ be the state transition probabilities of the Markov chains induced by $(\pi_i^k, P_i)$ and $(\pi_i^k,\hat{P}_i^k)$, respectively, i.e.,
\begin{align*}
&p_i^k(s'_i|s_i) = \sum_{a_i\in A_i}\pi_i^k(a_i|s_i)P_i(s'_i|a_i,s_i) \ \forall s_i,s'_i,\cr
&\hat{p}_i^k(s'_i|s_i) = \sum_{a_i\in A_i}\pi_i^k(a_i|s_i)\hat{P}_i^k(s'_i|a_i,s_i)\ \forall s_i,s'_i.
\end{align*}

In order to measure the proximity of the iterates generated by Algorithm \ref{alg:main} to an NE, we will use the well-known Nikaido-Isoda gap function \cite{nikaido1955note} as the ``merit" function.
\begin{definition}
	The Nikaido-Isoda gap function\footnote{We use bold symbols to denote aggregate variables of all the players, e.g., $\boldsymbol{\pi}^k = (\pi_1^k,\ldots,\pi_n^k)$, $\boldsymbol{q}^k = (q_1^k,\ldots,q_n^k)$, $\boldsymbol{\hq}^k = (\hq_1^k,\ldots,\hq_n^k)$, $\boldsymbol{P}= \prod_{i = 1}^nP_i$, $\boldsymbol{\Delta}= \prod_{i = 1}^n\Delta_i$, and $\boldsymbol{\Delta}_{\boldsymbol{\delta}}= \prod_{i = 1}^n\Delta_{i,\delta_i}$.} $\Psi:\boldsymbol{\Delta}(\bs{P})\times \boldsymbol{\Delta}\to \R$ is given by\footnote{Since every stationary policy $\bs{\pi}$ can be induced by some $\bs{q}\in \boldsymbol{\Delta}(\bs{P})$, we  define $\Psi$ on the domain $\boldsymbol{\Delta}(\bs{P})\times \boldsymbol{\Delta}$ instead of $\boldsymbol{\Delta}\times \boldsymbol{\Delta}$.}
	\begin{align*}
		\Psi(\bs{q}',\bs{q}) \triangleq \sum_{i=1}^n [V_i(q_i',q_{-i})-V_i(q_i,q_{-i})] = \sum_{i=1}^n [V_i(\pi_i^{q_i'},\pi_{-i}^{q_{-i}})-V_i(\pi_i^{q_i},\pi_{-i}^{q_{-i}})].
	\end{align*}
\end{definition}
The advantage of the Nikaido-Isoda gap function is that if  $\max_{\bs{q}' \in \boldsymbol{\Delta}(\bs{P})} \Psi(\bs{q'},\bs{q})\le \epsilon$, then 
\begin{align*} 	\max_{\pi_i'}V_i(\pi_i',\pi_{-i}^{q_{-i}})-V_i(\pi_i^{q_i},\pi_{-i}^{q_{-i}}) = \max_{q_i' \in \Delta_i(P_i)}V_i(q_i',q_{-i})-V_i(q_i,q_{-i})\le \epsilon\ \forall i, 
	 \end{align*}
which implies that $\bs{\pi}^{\bs{q}} = (\pi_1^{q_1},\ldots, \pi_n^{q_n})$ is an $\epsilon$-NE of the game $\G$. Moreover, for $\bs{\delta} = (\delta_1,\ldots,\delta_n)$ satisfying Lemma \ref{lem:shrunk} with $\epsilon/2$, if we also have $\max_{\bs{q}' \in \boldsymbol{\Delta}_{\boldsymbol{\delta}}(\bs{P})} \Psi(\bs{q'},\bs{q})\le \epsilon/2$, then 
	\begin{align*}
		\max_{\pi_i'}V_i(\pi_i',\pi_{-i}^{q_{-i}})-V_i(\pi_i^{q_i},\pi_{-i}^{q_{-i}}) &= \max_{q_i' \in \Delta_i(P_i)}V_i(q_i',q_{-i})-V_i(q_i,q_{-i})\cr
  &\le \max_{q_i' \in \Delta_{i,\delta_i}}V_i(q_i',q_{-i})-V_i(q_i,q_{-i})+\frac{\epsilon}{2}\le \epsilon \ \forall i,
	\end{align*}
	which shows that $\bs{\pi}^{\bs{q}}$ is an $\epsilon$-NE of the stochastic game $\G$. This suggests that the value of $\max_{\bs{q}' \in \boldsymbol{\Delta}_{\boldsymbol{\delta}}(\bs{P})} \Psi(\bs{q'},\bs{q})$ can serve as a measure of proximity of $\bs{q}$ (or $\bs{\pi}^{\bs{q}}$) to an $\epsilon$-NE of the game $\G$. 
 
Now, we are ready to state the main result of this section, which provides a polynomial-time convergence rate for Algorithm \ref{alg:main} in terms of the average Nikaido-Isoda gap function. 
\begin{theorem}
	\label{theo:finite}
	Assume Assumptions \ref{ass:indep}, \ref{ass:ergodic}, and \ref{ass:minprob} hold. Given any $\gamma\in (0, 1)$, if each player $i$ follows Algorithm \ref{alg:main} with a proper choice of stepsize $\eta^k = \frac{c}{\sqrt{K}}\ \forall k$, quadratic regularizer $h(\cdot) = \frac{1}{2}\|\cdot\|^2$, and parameters 
 \begin{align}\nonumber
      d= \tau \ln(\frac{(1-e^{-1/\tau})\sqrt{K}}{2\min_i|S_i|}),  \ \ \ \ \epsilon_i^k(s'_i|s_i,a_i) = \sqrt{\frac{\ln(nK|A_i||S_i|^2)-\ln(\gamma)}{\max(1,N_i^k(s_i,a_i))}},
 \end{align}
then with probability at least $1-2\gamma$, the sequence $\{\bs{\hq}^k\}_{k=1}^K$ satisfies
	\begin{align}
		\label{eq:finite_bound}
\max_{\bs{q}'\in\bs{\Delta}_{\bs{\delta}}(\bs{P})}\sum_{k=1}^{K}\frac{\eta^k}{w^K}\Psi(\bs{q}',\bs{\hq}^k)\le	\O\Big(\frac{\sum_{i=1}^n|S_i|\sqrt{\ln(nK|A_i||S_i|^2)-\ln\gamma}}{(1-e^{-1/\tau})\sqrt{K}}+\frac{\sum_{i=1}^n|A_i||S_i|^2}{\sqrt{K}}\Big),
	\end{align}
where $w_K=\sum_{k=1}^{K}\eta^k$. In particular, we have $\max_{\bs{q}'\in\bs{\Delta}_{\bs{\delta}}(\bs{P})}\sum_{k=1}^{K}\frac{\eta^k}{w^K}\Psi(\bs{q}',\bs{\hq}^k)\le \epsilon$ for any $K$ that satisfies
	$$K = \tilde{\O}\left(\max\Big\{\frac{\big(\sum_{i=1}^n|S_i|\sqrt{\ln(n|A_i||S_i|^2)-\ln\gamma}\big)^2}{(1-e^{-1/\tau})^2\epsilon^2},\frac{\big(\sum_{i=1}^n|A_i||S_i|^2\big)^2}{\epsilon^2}\Big\}\right).$$
\end{theorem}


\begin{proof}
First, we note that
	\begin{align}\label{eq:finite_decompose}
	\max_{\bs{q}'\in\bs{\Delta}_{\bs{\delta}}(\bs{P})}\sum_{k=1}^{K}\frac{\eta^k}{w^K}\Psi(\bs{q}',\bs{\hq}^k)
	& = \sum_{i=1}^n\Big(\max_{q_i'\in \Delta_{i,\delta_i}(P_i)}\Big\{\sum_{k=1}^{K}\frac{\eta^k}{w^K}[V_i(\pi_i^{q_i'},\pi_{-i}^{\hq_{-i}^k})-V_i(\pi_i^{\hq_i^k},\pi_{-i}^{\hq_{-i}^k})]\Big\}\Big)\cr
	&=\sum_{i=1}^n\Big(\max_{q_i'\in \Delta_{i,\delta_i}(P_i)}\Big\{\sum_{k=1}^{K}\frac{\eta^k}{w^K}[V_i(\pi_i^{q_i'},\pi_{-i}^k)-V_i(\pi_i^k,\pi_{-i}^k)]\Big\}\Big). 
	\end{align}
	For any $i\in [n]$, let $\pi_i^* :=  \pi_i^{\hat{q}^*_i}$, where 
	$$\hat{q}^*_i = \argmax_{q_i'\in \Delta_{i,\delta_i}(P_i)}\sum_{k=1}^{K}\frac{\eta^k}{w^K}[V_i(\pi_i^{q_i'},\pi_{-i}^k)-V_i(\pi_i^k,\pi_{-i}^k)].$$
Let us denote by $\nu_i^*$, $\rho_i^*$, and $q_i^*$, the occupancy measures induced by $\pi_i^*$ and the unknown transition probability matrix $P_i$ (note that since $q^*_i \in \Delta_{i,\delta_i}(P_i)$, from Lemma \ref{lem:feasible} we have $\hq^*_i= q_i^*$ ). Using \eqref{eq:linear-occupation-form1} in Proposition \ref{prop:dual_formulation}, we have
	\begin{align*}
		V_i(\pi_i^{\hat{q}^*_i},\pi_{-i}^k) - V_i(\pi_{i}^k,\pi_{-i}^k) = V_i(\pi_i^*,\pi_{-i}^k)-V_i(\pi_i^k,\pi_{-i}^k) = \langle \rho_i^*, v_i(\rho_{-i}^k) \rangle - \langle \rho_i^k, v_i(\rho_{-i}^k) \rangle.
	\end{align*}
	Therefore, we have
	\begin{align*}
		\sum_{k=1}^{K}\frac{\eta^k}{w^K}[V_i(\pi_i^{\hq_i^*},\pi_{-i}^k)&-V_i(\pi_i^k,\pi_{-i}^k)]=\sum_{k=1}^{K}\frac{\eta^k}{w^K}\Big(\langle \rho_i^*, v_i(\rho_{-i}^k) \rangle - \langle \rho_i^k, v_i(\rho_{-i}^k) \rangle\Big)\\
		=& \underbrace{\sum_{k=1}^{K}\frac{\eta^k}{w^K}\langle \hat{\rho}_i^k-\rho_i^k, v_i(\rho_{-i}^k) \rangle}_{\text{Error}}
		+\underbrace{\sum_{k=1}^{K}\frac{\eta^k}{w^K}\langle \rho_i^*-\hat{\rho}_i^k, R_i^k \rangle}_{\text{Regret}}
		+\underbrace{\sum_{k=1}^{K}\frac{\eta^k}{w^K}\langle \rho_i^*-\hat{\rho}_i^k, v_i(\rho_{-i}^k)-R_i^k \rangle}_{\text{Bias}}.
	\end{align*}
	Here, the first term ``Error" measures the error between $\hat{\rho}_i^k$ and $\rho_i^k$, the second term ``Regret" is the intrinsic regret of the OMD algorithm for the online linear optimization problem, and the third term ``Bias" corresponds to the bias of using $R_i^k$ to approximate $v_i(\rho_{-i}^k)$. Therefore, in order to show \eqref{eq:finite_bound}, we only need to provide high probability bounds for each of the three terms: Error, Regret, and Bias. To that end, we first show in the following lemma that the event $\{P_i \in \mathcal{P}_i^k, \forall k \in [K]\text{, } i \in [n]\}$ happens with probability at least $1-\gamma$. 
	\begin{restatable}{lemma}{confidenta}
		\label{lem:confident1}
		Assume that each player $i$ follows Algorithm \ref{alg:main} with the choice of $\epsilon_i^k(s'_i|s_i,a_i) = \sqrt{\frac{\ln(nK|A_i||S_i|^2)-\ln(\gamma)}{2\max(1,N_i^k(s_i,a_i))}}$. Then, with probability at least $1-\gamma$, we have $P_i \in \mathcal{P}_i^k\ \forall k \in [K]\text{, } i \in [n]$. Moreover, under this event, we have
		\begin{align}
			\max_{s_i,a_i,s'_i}\Big|\hat{P}_i^k(s'_i|s_i,a_i)-P_i(s'_i|s_i,a_i)\Big| \le 2\max_{s_i,a_i,s'_i}\epsilon_i^k(s_i,a_i,s'_i)\le \sqrt{\frac{2\ln(nK|A_i||S_i|^2)-\ln(\gamma)}{k}}.\label{eq:maxepsilon}
		\end{align}
	\end{restatable}
	Next, conditioned on the event $\{P_i \in \mathcal{P}_i^k, \forall k \in [K]\text{, } i \in [n]\}$, we can bound each of the terms Error, Regret, and Bias, by using the following lemmas whose proofs are given in Appendix A.
	\begin{restatable}{lemma}{error}
		\label{lem:error}
		Under the assumptions of Theorem \ref{theo:finite} and conditioned on the event $\{P_i \in \mathcal{P}_i^k\  \forall k \in [K], i \in [n]\}$,
		\begin{align*}
			\text{Error}\le \frac{|S_i|\sqrt{2\ln(nK||A_i|S_i|^2)-\ln\gamma}}{(1-e^{-1/\tau})w^K}\sum_{k=1}^{K}\frac{\eta^k}{\sqrt{k}}.
		\end{align*}
	\end{restatable}

	\begin{restatable}{lemma}{regret}
		\label{lem:regret}
		Under the assumptions of Theorem \ref{theo:finite} and conditioned on the event $\{P_i \in \mathcal{P}_i^k\ \forall k \in [K], i \in [n]\}$,
		\begin{align*}
			\text{Regret} \le \frac{|A_i||S_i|^2}{2\mu w^K}\sum_{k=1}^{K}(\eta^k)^2+\frac{D_{h_i}(q_i^*||\hq_i^1)}{w^K}.
		\end{align*}
	\end{restatable}

	\begin{restatable}{lemma}{bias}
		\label{lem:bias}
		Under the assumptions of Theorem \ref{theo:finite}, with probability at least $1-\gamma/n$, we have
		\begin{align*}
			\text{Bias}\le \frac{2}{w^K}\sqrt{2\ln\frac{n}{\gamma}\sum_{k=1}^K(\eta^k)^2} +2e^{-\frac{d}{\tau}}.
		\end{align*}
	\end{restatable}
	Putting everything together, we have shown that conditioned on the event $\{P_i \in \mathcal{P}_i^k\ \forall k \in [K], i \in [n]\}$, with probability at least $1-\gamma/n$,
	\begin{align*}
		\sum_{k=1}^{K}\frac{\eta^k}{w^K}[V_i(\pi_i^{\hq_i^*},\pi_{-i}^k)-V_i(\pi_i^k,\pi_{-i}^k)]   \le& \frac{|S_i|\sqrt{2\ln(\frac{nK|A_i||S_i|^2}{\gamma})}}{(1-e^{-1/\tau})w^K}\sum_{k=1}^{K}\frac{\eta^k}{\sqrt{k}}+\frac{|A_i||S_i|^2}{2\mu w^K}\!\sum_{k=1}^{K}\!(\eta^k)^2+\frac{D_{h_i}(q_i^*||\hq_i^1)}{w^K}\\
		&+\frac{2}{w^K}\sqrt{2\ln\frac{n}{\gamma}\sum_{k=1}^K(\eta^k)^2} +2e^{-\frac{d}{\tau}}.
	\end{align*}
	If we take $\eta^k = \frac{c}{\sqrt{K}}, \forall k, h(\cdot) = \frac{1}{2}\|\cdot\|^2 \text{ and } d= \tau \ln(\frac{(1-e^{-1/\tau})\sqrt{K}}{2\min_i|S_i|})$, then we have $w_K = c\sqrt{K}$. Moreover, since $\sum_{k=1}^{K}\frac{1}{\sqrt{k}}\le 2\sqrt{K}$ and $D_{h_i}(q_i^*||\hq_i^1) = \|q_i^*-\hq_i^1\|^2\le |A_i||S_i|^2$, we can write
	\begin{align}
		\sum_{k=1}^{K}\frac{\eta^k}{w^K}[V_i(\pi_i^{\hq_i^*},\pi_{-i}^k)-V_i(\pi_i^k,\pi_{-i}^k)]   \le&\frac{2|S_i|\sqrt{2\ln(\frac{nK|A_i||S_i|^2}{\gamma})}}{(1-e^{-1/\tau})\sqrt{K}}
		+ \frac{c|A_i||S_i|^2}{2\mu \sqrt{K}}+\frac{|A_i||S_i|^2}{c\sqrt{K}}\cr
		&+\frac{2\sqrt{2\ln\frac{n}{\gamma}}}{\sqrt{K}}
		+\frac{2|S_i|}{(1-e^{-1/\tau})\sqrt{K}}\cr
		&=\O\Big(\frac{|S_i|\sqrt{\ln\frac{nK|A_i||S_i|^2}{\gamma}}}{(1-e^{-1/\tau})\sqrt{K}}+\frac{|A_i||S_i|^2}{\sqrt{K}}\Big).\label{eq:finite_O}
	\end{align}
	Using the union bound we have that conditioned on the event $\{P_i \in \mathcal{P}_i^k\ \forall k \in [K], i \in [n]\}$, with probability at least $1-\gamma$, \eqref{eq:finite_O} holds for all $i\in [n]$. Thus, using Lemma \ref{lem:confident1} and \eqref{eq:finite_decompose}, with probability at least $1-2\gamma$, we have
	\begin{align*}
		\max_{\bs{q}'\in\bs{\Delta}_{\bs{\delta}}(\bs{P})}\sum_{k=1}^{K}\frac{\eta^k}{w^K}\Psi(\bs{q}',\bs{\hq}^k)\le 
		\O\Big(\frac{\sum_{i=1}^n|S_i|\sqrt{\ln(nK|A_i||S_i|^2)-\ln\gamma}}{(1-e^{-1/\tau})\sqrt{K}}+\frac{\sum_{i=1}^n|A_i||S_i|^2}{\sqrt{K}}\Big).
	\end{align*}
	Finally, using the above expression, it is easy to see that in order to have $\max_{\bs{q}'\in\bs{\Delta}_{\bs{\delta}}(\bs{P})}\sum_{k=1}^{K}\frac{\eta^k}{w^K}\Psi(\bs{q}',\bs{\hq}^k)\le \epsilon$, it suffices to take $$K = \tilde{\O}\Big(\max\Big\{\frac{\big(\sum_{i=1}^n|S_i|\sqrt{\ln(n|A_i||S_i|^2)-\ln\gamma}\big)^2}{(1-e^{-1/\tau})^2\epsilon^2},\frac{\big(\sum_{i=1}^n|A_i||S_i|^2\big)^2}{\epsilon^2}\Big\}\Big).$$

\end{proof}

\section{Asymptotic Convergence to an $\epsilon$-Nash Equilibrium Policy}
\label{sec:asymptotic}
In this section, we show that if Algorithm \ref{alg:main} is run with the choice of $\bs{\delta} = (\delta_1,\ldots,\delta_n)$ satisfying Lemma \ref{lem:shrunk}, then the iterates generated by Algorithm \ref{alg:main} will converge asymptotically to a \emph{globally stable} $\epsilon$-NE policy of the game $\G$ (if it exists) with high probability. We begin with the following definition.
\begin{definition}
	Given $\bs{\delta} = (\delta_1,\ldots,\delta_n)$, we define $\V_{\bs{\delta}}$ to be the constrained version of the virtual game $\V$ in which the action set for each player $i$ is given by $\Delta_{i,\delta_i} (P_i)$ (instead of $\Delta_i (P_i)$).
\end{definition}

From Lemma \ref{lem:shrunk}, we know that a NE of $\V_{\bs{\delta}}$ is an $\epsilon$-NE of the game $\V$, and so its induced policy is an $\epsilon$-NE of the original stochastic game $\G$. In order to establish asymptotic convergence of Algorithm \ref{alg:main} to an $epsilon$-NE policy, we make the following assumption on the constrained virtual game $\V_{\bs{\delta}}$.
\begin{assumption}(\!\cite{mertikopoulos2019learning})
	\label{ass:g_stable}
	The constrained virtual game $\V_{\bs{\delta}}$ admits a unique NE $\boldsymbol{q}^*$ that is \emph{globally stable}, i.e.,
	\begin{align}
		\label{eq:g_stable1}
		\la \boldsymbol{v}(\boldsymbol{q}),\boldsymbol{q}^*-\boldsymbol{q}\ra \ge 0, \qquad \forall \boldsymbol{q}\in \boldsymbol{\Delta_{\bs{\delta}}}(\boldsymbol{P}),
	\end{align}
	with equality if and only if $\boldsymbol{q} = \boldsymbol{q}^*$, where $\boldsymbol{v}(\boldsymbol{q}) = (v_i(q_{-i}),i \in [n])$ is the vector of players’ payoff gradients with respect to their own strategies. From Definition \ref{def:v_game}, the global stability condition \eqref{eq:g_stable1} can be  written as
	\begin{align}
	\label{eq:g_stable2}
	\sum_{i=1}^n \la l_i(q_{-i}),q_i^*-q_i\ra \ge 0, \qquad \forall \boldsymbol{q}\in \boldsymbol{\Delta_{\bs{\delta}}}(\boldsymbol{P}).
	\end{align}
\end{assumption}

The notion of \emph{variational stability} was first introduced in \cite{mertikopoulos2019learning} as a relaxation of the well-known \emph{monotonicity condition} \cite{rosen1965existence} for static continuous action games. Specifically, while requiring that the game admits a unique NE (which is almost a prerequisite for any learning algorithm to converge), it is shown in \cite{mertikopoulos2019learning} that the monotonicity condition implies that the game admits a unique NE that is \emph{globally stable}. We are now ready to state the main result of this section.
\begin{theorem}
	\label{theo:assymptotic}
	Assume Assumptions \ref{ass:indep}, \ref{ass:ergodic}, and  \ref{ass:g_stable} hold. Given any $\epsilon>0$, assume each player $i$ follows Algorithm \ref{alg:main} with a choice of  $\delta_i$ satisfying Lemma \ref{lem:shrunk}, and a nonincreasing sequence of step-sizes satisfying $\sum_{k=1}^{\infty}\eta^k = \infty$, $\sum_{k=1}^{\infty}(\eta^k)^2 < \infty $, and $\sum_{k=1}^{\infty}\eta^k\sqrt{\ln k/k} < \infty $. Then, for any $\gamma\in (0,1)$ and the choice of parameters $d^k = 2\tau\ln k$ and $\epsilon_i^k(s'_i|s_i,a_i) = \sqrt{\frac{\ln(2nk^2|A_i||S_i|^2)-\ln\gamma}{2\max(1,N_i^k(s_i,a_i))}}$, we have $\lim_{k\to\infty}\boldsymbol{\pi}^k = \boldsymbol{\pi}^*$ with probability at least $1-\gamma$, where $\boldsymbol{\pi}^*$ is an $\epsilon$-NE of the game $\G$.
\end{theorem}

	To prove Theorem \ref{theo:assymptotic}, we first state the following lemmas, whose proofs are deferred to Appendix B. Lemma \eqref{lem:confident2} is an analogous version of Lemma \eqref{lem:confident1} except that it extends the range for $k\in [K]$ to any nonnegative integer  $k\in \N$. Moreover, Lemma \ref{lem:qhat_converge} shows the convergence of $\boldsymbol{\hat{q}}^k$ to $\boldsymbol{q}^k$ as $k\to \infty$.   

        \begin{restatable}{lemma}{confidentb}
		\label{lem:confident2}
		Under the same assumptions as in Theorem \ref{theo:assymptotic}, with probability at least $1-\gamma$, we have $P_i \in \mathcal{P}_i^k, \forall k \in \N\text{, } i \in [n]$. Moreover, under this event, we have
		\begin{align}	\max_{s_i,a_i,s'_i}\Big|\hat{P}_i^k(s'_i|s_i,a_i)-P_i(s'_i|s_i,a_i)\Big| \le 2\max_{s_i,a_i,s'_i}\epsilon_i^k(s_i,a_i,s'_i)\le \sqrt{\frac{2\ln (2nk^2|A_i||S_i|^2)-\ln\gamma}{k}}.\label{eq:maxepsilon2}
		\end{align}
	\end{restatable}

	\begin{restatable}{lemma}{qhatconverge}
		\label{lem:qhat_converge}
		Under the same assumptions as in Theorem \ref{theo:assymptotic}, and conditioned on the event $\{P_i \in \mathcal{P}_i^k\  \forall k \in \N, i \in [n]\}$, we have $\lim_{k\to\infty} \boldsymbol{q}^k -  \boldsymbol{\hq}^k = \mathbf{0}$.
	\end{restatable}
	
\medskip	
\begin{proof}{[Proof of Theorem \ref{theo:assymptotic}]}
	We first note that it suffices to show with probability at least $1-2\gamma$, we have 
	\begin{align*}
		\lim_{k\to\infty}\boldsymbol{q}^k = \boldsymbol{q}^*,
	\end{align*}
where $\boldsymbol{q}^*$ is the unique stable Nash equilibrium of the constrained virtual game $\V_{\bs{\delta}}$. Let us define
\begin{align*}
	\boldsymbol{D}_{\boldsymbol{h}}(\boldsymbol{q}^*||\boldsymbol{\hq}^k)\triangleq \sum_{i=1}^n D_{h_i}(q_i^*||\hq_i^k).
\end{align*}
From the definition of Bregman divergence as well as Lemma \ref{lem:qhat_converge}, we have
\begin{align}
	\lim_{k\to\infty}\boldsymbol{q}^k = \boldsymbol{q}^* \Leftrightarrow \lim_{k\to\infty}\boldsymbol{\hq}^k = \boldsymbol{q}^* \Leftrightarrow  \lim_{k\to\infty}\boldsymbol{D}_{\boldsymbol{h}}(\boldsymbol{q}^*||\boldsymbol{\hq}^k) = 0.\label{eq:equivelent}
\end{align}
The rest of the proof proceeds in two steps. In the first step, we show that every neighborhood $U\subset \boldsymbol{\Delta}(\boldsymbol{P})$ of $\boldsymbol{q}^*$ is recurrent in $\{\boldsymbol{q}^k\}_{k=1}^{\infty}$. This is shown in the following lemma.
\begin{restatable}{lemma}{recurrent}
	\label{lem:recurrent}
	Under the same assumptions as in Theorem \ref{theo:assymptotic}, and conditioned on the event $\{P_i \in \mathcal{P}_i^k\ \forall k \in \N, i \in [n]\}$, every open neighborhood $U\subset \boldsymbol{\Delta}(\boldsymbol{P})$ of $\boldsymbol{q}^*$ is recurrent in $\{\boldsymbol{q}^k\}_{k=1}^{\infty}$. More specifically, there exists a subsequence $\boldsymbol{q}^{k_m}$ of $\boldsymbol{q}^k$ such that $\boldsymbol{q}^{k_m}\to \boldsymbol{q}^*$ almost surely.
\end{restatable}

 Equipped with Lemma \ref{lem:recurrent}, in the second step we further show that for any $\epsilon, \delta>0$, there exist $k_0\in \N$, such that $\P(\boldsymbol{D}_{\boldsymbol{h}}(\boldsymbol{q}^*||\boldsymbol{\hq}^k) \le \epsilon,\forall k>k_0)\ge 1-\delta$. More precisely, we have:
\begin{restatable}{lemma}{stayin}
	\label{lem:stay_in}
	Under the same assumptions as in Theorem \ref{theo:assymptotic}, and conditioned on the event $\{P_i \in \mathcal{P}_i^k\ \forall k \in \N, i \in [n]\}$, for any $\epsilon, \delta>0$, there exists $k_0\in \N$ such that $\P(\boldsymbol{D}_{\boldsymbol{h}}(\boldsymbol{q}^*||\boldsymbol{\hq}^k) \le \epsilon,\forall k>k_0)\ge 1-\delta.$
\end{restatable}
To complete the proof, let $E_{\epsilon}$ denote the event $\{\exists k_0\in \N,\ s.t.\ \boldsymbol{D}_{\boldsymbol{h}}(\boldsymbol{q}^*||\boldsymbol{\hq}^k) \le \epsilon,\ \forall k>k_0\}$. Then, we have $\P(\lim_{k\to\infty}\boldsymbol{D}_{\boldsymbol{h}}(\boldsymbol{q}^*||\boldsymbol{\hq}^k) = 0) = \P(\cap_{r = 1}^{\infty}E_{2^{-r}})$. Using Lemma \ref{lem:stay_in} and conditioned on the event $\{P_i \in \mathcal{P}_i^k\ \forall k \in \N, i \in [n]\}$, we can write
\begin{align*}
	&\forall \epsilon,\delta>0,\quad \P(\exists k_0\in \N,\ s.t.\ \boldsymbol{D}_{\boldsymbol{h}}(\boldsymbol{q}^*||\boldsymbol{\hq}^k) \le \epsilon,\ \forall k>k_0)\ge 1-\delta\\
	&\Rightarrow\forall \epsilon>0,\quad \P(\exists k_0\in \N,\ s.t.\ \boldsymbol{D}_{\boldsymbol{h}}(\boldsymbol{q}^*||\boldsymbol{\hq}^k) \le \epsilon,\ \forall k>k_0)= 1\\
	&\Rightarrow\forall \epsilon>0,\quad \P(E_{\epsilon}) = 1\\
	& \Rightarrow P(\lim_{k\to\infty}\boldsymbol{D}_{\boldsymbol{h}}(\boldsymbol{q}^*||\boldsymbol{\hq}^k) = 0)  = \P(\cap_{r = 1}^{\infty}E_{2^{-r}}) = 1.
\end{align*}
Therefore, using Lemma \ref{lem:confident1} and \eqref{eq:equivelent}, we conclude that with probability at least $1-\gamma$, we have 
\begin{align*}
	\lim_{k\to\infty}\boldsymbol{q}^k = \boldsymbol{q}^*.
\end{align*}
	
\end{proof}

\section{Conclusion}
\label{sec:conclusion}
In this work, we studied the class of stochastic games with independent chains and unknown transition matrices. Relying on a compact dual formulation of the game based on occupancy measures and the technique of confidence set to maintain high-probability estimates of the unknown transition matrices, we proposed a fully decentralized online mirror descent algorithm to learn an $\epsilon$-NE for this class of stochastic games. The proposed algorithm has the desired properties of independence, scalability, and convergence. Specifically, under no assumptions on the reward functions, we showed the proposed algorithm converges in polynomial time in a weaker distance (namely, the averaged Nikaido-Isoda gap) to the set of $\epsilon$-NE policies with arbitrarily high probability. Under the variational stability assumption of the game, we showed that the proposed algorithm converges asymptotically to an $\epsilon$-NE with arbitrarily high probability.

As a future research direction, one can consider relaxing the requirement that the algorithm works on the shrunk polytope as defined in Definition \ref{def:shrunk}. This requirement is essential to encourage exploration, and we believe it can be further removed using techniques such as implicit exploration \cite{neu2015explore}. Another interesting problem is improving the sampling strategy of the algorithm used to construct $R_i^k$, which is an estimator of $l^k_i$. While the proposed algorithm requires every player to visit each state-action pair during each episode, advanced sampling techniques may remove this requirement, thus improving the algorithm's performance. Finally, another interesting research direction is identifying other classes of $n$-player stochastic games beyond zero-sum or identical interest stochastic games (e.g., \cite{sayin2023decentralized}) that admit efficient learning algorithms for obtaining their equilibrium points.

\bibliographystyle{IEEEtran}
\bibliography{thesisrefs}

\medskip
\section*{Appendix A: Proof of Lemmas For Theorem \ref{theo:finite}}
\label{subsec:appendixa}

\medskip
Following the definitions of variables $\nu_i^k, \hat{\nu}_i^k, \rho_i^k, \hat{\rho}_i^k, q_i^k, \hat{q}_i^k$, one can immediately see that
\begin{align*}
	\nu_i^k = \nu_i^k p_i^k,\quad &\hat{\nu}_i^k = \hat{\nu}_i^k \hat{p}_i^k,\\
	\rho_i^k(s_i,a_i) = \nu_i^k(s_i) \pi_i^k(a_i|s_i),\quad &\hat{\rho}_i^k(s_i,a_i) = \hat{\nu}_i^k(s_i) \pi_i^k(a_i|s_i) \quad\forall s_i,a_i,\\
	q_i^k(s_i,a_i,s'_i) = \rho_i^k(s_i,a_i) P_i(s'_i|s_i,a_i),\quad &\hat{q}_i^k(s_i,a_i,s'_i) = \hat{\rho}_i^k(s_i,a_i) \hat{P}_i^k(s'_i|s_i,a_i) \quad\forall s_i,a_i,s'_i.
\end{align*}
From Assumption \ref{ass:ergodic}, we also have 
\begin{align}
	\|(\nu-\nu')p_i^k\|_1\le e^{-1/\tau}\|\nu-\nu'\|_1\qquad \forall \nu,\nu'\in\Delta(S_i). \label{eq:mix}
\end{align}

We first establish the following lemma, which states that $R_i^k$ is an (almost) unbiased estimator for $v_i(\rho_{-i}^k)$. Let $\F_{k}$ be the filtration adapted to all measurable events up to the start of spisode $k$. Then, we have:
\begin{lemma}(\!\cite{etesami2022learning}, Lemma 3)
	\label{lem:Rbias}
	Let Assumption \ref{ass:ergodic} hold and assume that each player $i$ follows Algorithm \ref{alg:main}. Conditioned on $\F_{k}$, the expectation of  reward vector $R_i^k$ that player $i$ computes at the end of episode $k$ satisfies
	\begin{align}
		\label{eq:Rbias_1}
		\E_{k}[R_i^k]-v_i(\rho_{-i}^k) \le e^{-\frac{d}{\tau}}\mathbf{1},
	\end{align}
	where the above inequality is coordinatewise, and the expectation is with respect to the randomness of players’ policies and their internal chains. If with some abuse of notations, we define $R_i^k:S_i\times A_i \times S_i\to [0,1]$, where $R_i^k(s_i,a_i,s'_i) = R_i^k(s_i,a_i)$, then from Proposition \ref{prop:dual_formulation}, we also have
	\begin{align}
		\label{eq:Rbias_2}
		\E_{k}[R_i^k]-l_i(q_{-i}^k) \le e^{-\frac{d}{\tau}}\mathbf{1}.
	\end{align}
\end{lemma}

\medskip
The intuition behind Lemma \ref{lem:Rbias} is that if the sampling distribution of $R_i^k$ is the stationary distribution $\rho_i^k$, then $\E_{k}[R_i^k] = v_i(\rho_{-i}^k)$. While this is not the case, due to Assumption \ref{ass:ergodic}, the ``warm-up" phase of length $d$ will drive the sampling distribution of $R_i^k$ close to $\rho_i^k$ at an exponential rate, resulting in \eqref{eq:Rbias_1}. We also need the following well-known proposition for Bregman divergence.

\begin{proposition}[\!\cite{chen1993convergence}]
	Let $D_h$ be the Bregman divergence with respect to $h$. Then, for any three points ${ x},{ y} ,{ z}$, the following identity holds
	\begin{align}
		\label{eq:Bregman_identity}
		D_h({ z}|| { x}) + D_h({ x}||{ y}) - D_h({ z}|| { y}) = \langle \nabla h({ y}) - \nabla h({ x}), { z} - { x}\rangle~. 
	\end{align}
\end{proposition}

\confidenta*
\begin{proof}
	From Hoeffding’s inequality, with probability at least $1-\frac{\delta}{nK|A_i||S_i|^2}$ we have
	\begin{align*}
		\bar{P}_i^k(s'_i|s_i,a_i)-P_i(s'_i|s_i,a_i)\le \sqrt{\frac{\ln (nK|A_i||S_i|^2)-\ln \gamma}{2\max(1,N_i^k(s_i,a_i))}} = \epsilon_i^k(s_i,a_i,s'_i).
	\end{align*}
	Using union bound, with probability at least $1-\gamma$, we have that the above bound hold for all $(s_i,a_i,s'_i)\in S_i\times A_i \times S_i$, $k\in [K], i\in [n]$, which implies that $P_i \in \mathcal{P}_i^k\ \forall k \in [K], i \in [n]$. Moreover, since from the construction of Algorithm \ref{alg:main}, each player $i$ will visit all its state-action pairs $(s_i,a_i)$ at least once, we must have $N_i^k(s_i,a_i)\ge k\ \forall(s_i,a_i)\in S_i\times A_i, k\in [K], i\in [n]$. Thus, we get $\max\limits_{s_i,a_i,s'_i}\epsilon_i^k(s_i,a_i,s'_i)\le \sqrt{\frac{\ln(nK|A_i||S_i|^2)-\ln\gamma}{2k}}$.
\end{proof}

\error*
\begin{proof}
	Since $r_i\in [0,1]\ \forall i$, we have $v_i(\rho_{-i}^k)_{(s_i,a_i)}\in [0,1]\ \forall s_i,a_i$. Therefore,
	\begin{align*}
		\langle \hat{\rho}_i^k-\rho_i^k, v_i(\rho_{-i}^k) \rangle &\le \sum_{s_i,a_i}|\hat{\rho}_i^k(s_i,a_i)-\rho_i^k(s_i,a_i)|\\
		&=\sum_{s_i,a_i}|\hat{\nu}_i^k(s_i)\pi_i^k(a_i|s_i)-\nu_i^k(s_i)\pi_i^k(a_i|s_i)|\\
		&=\sum_{s_i}|\hat{\nu}_i^k(s_i)-\nu_i^k(s_i)|\sum_{a_i}\pi_i^k(a_i|s_i)\\
		&=\|\hat{\nu}_i^k-\nu_i^k\|_1.
	\end{align*}
	From the definitions of $\hat{\nu}_i^k$ and $\nu_i^k$, we have
	\begin{align*}
		\|\hat{\nu}_i^k-\nu_i^k\|_1 &= \|\hat{\nu}_i^k\hat{p}_i^k-\nu_i^kp_i^k\|_1\\
		&=\|\hat{\nu}_i^k\hat{p}_i^k-\hat{\nu}_i^kp_i^k+\hat{\nu}_i^kp_i^k-\nu_i^kp_i^k\|_1\\
		&\le \|\hat{\nu}_i^k\hat{p}_i^k-\hat{\nu}_i^kp_i^k\|_1+\|\hat{\nu}_i^kp_i^k-\nu_i^kp_i^k\|_1\\
		&\le\|\hat{\nu}_i^k\hat{p}_i^k-\hat{\nu}_i^kp_i^k\|_1+e^{-1/\tau}\|\hat{\nu}_i^k-\nu_i\|_1,\\
	\end{align*}
	where the last inequality follows from \eqref{eq:mix}. Therefore,
	\begin{align*}
		\|\hat{\nu}_i^k-\nu_i^k\|_1  \le \frac{1}{1-e^{-1/\tau}}\|\hat{\nu}_i^k(\hat{p}_i^k-p_i^k)\|_1.
	\end{align*}
	Now, conditioned on the event $\{P_i \in \mathcal{P}_i^k\ \forall k \in [K], i \in [n]\}$, we can upper-bound $\|\hat{\nu}_i^k(\hat{p}_i^k-p_i^k)\|_1$ as
	\begin{align*}
		\|\hat{\nu}_i^k(\hat{p}_i^k-p_i^k)\|_1& = \sum_{s'_i}\Big|\sum_{s_i}\hat{\nu}_i^k(s_i)(\hat{p}_i^k(s'_i|s_i)-p_i^k(s'_i|s_i))\Big|\\
		&=\sum_{s'_i}\Big|\sum_{s_i}\hat{\nu}_i^k(s_i)\Big(\sum_{a_i}\pi_i^k(a_i|s_i)(\hat{P}_i^k(s'_i|s_i,a_i)-P_i(s'_i|s_i,a_i))\Big)\Big|\\
		&\le \sum_{s'_i}\sum_{s_i}\hat{\nu}_i^k(s_i)\Big|\sum_{a_i}\pi_i^k(a_i|s_i)(\hat{P}_i^k(s'_i|s_i,a_i)-P_i(s'_i|s_i,a_i))\Big|\\
		&\le \sum_{s'_i}\sum_{s_i}\hat{\nu}_i^k(s_i)\max_{a_i}\Big|\hat{P}_i^k(s'_i|s_i,a_i)-P_i(s'_i|s_i,a_i)\Big|\\
		&\le \sum_{s'_i}\max_{a_i,s_i}\Big|\hat{P}_i^k(s'_i|s_i,a_i)-P_i(s'_i|s_i,a_i)\Big|\\
		&\le |S_i|\max_{s_i,a_i,s'_i}\Big|\hat{P}_i^k(s'_i|s_i,a_i)-P_i(s'_i|s_i,a_i)\Big|\\
	&\le 2|S_i|\max_{s_i,a_i,s'_i}\epsilon_i^k(s_i,a_i,s'_i),
	\end{align*}
	where the last inequality follows from \eqref{eq:maxepsilon}. Thus, we can write 
	\begin{align}
		\label{eq:niu_1norm}
		\|\hat{\nu}_i^k-\nu_i^k\|_1  \le \frac{2|S_i|}{1-e^{-1/\tau}}\max_{s_i,a_i,s'_i}\epsilon_i^k(s_i,a_i,s'_i).
	\end{align}
	Putting everything together and using $\max_{s_i,a_i,s'_i}\epsilon_i^k(s_i,a_i,s'_i)\le \sqrt{\frac{\ln(nK|A_i||S_i|^2)-\ln \gamma}{2k}}$, we have
	\begin{align*}
		\text{Error}&\le \sum_{k=1}^{K}\frac{\eta^k}{w^K}\|\hat{\nu}_i^k-\nu_i^k\|_1\\
		&\le\sum_{k=1}^{K}\frac{\eta^k}{w^K}\frac{2|S_i|}{1-e^{-1/\tau}}\sqrt{\frac{\ln(nK|A_i||S_i|^2)-\ln \gamma}{2k}}\\
		&=\frac{|S_i|\sqrt{2\ln(nK|A_i||S_i|^2)-\ln\gamma}}{(1-e^{-1/\tau})w^K}\sum_{k=1}^{K}\frac{\eta^k}{\sqrt{k}}.
	\end{align*}
\end{proof}

\regret*
\begin{proof}
	With some abuse of notations, we define $R_i^k:S_i\times A_i \times S_i\to [0,1]$, where $R_i^k(s_i,a_i,s'_i) = R_i^k(s_i,a_i)$. Since $\rho_i^*(s_i,a_i) = \sum_{s'_i}q_i^*(s_i,a_i,s'_i), \hat{\rho}_i^*(s_i,a_i) = \sum_{s'_i}\hat{q}_i^*(s_i,a_i,s'_i),\forall a_i,s_i$, we can write 
	\begin{align*}
		\text{Regret} & = \sum_{k=1}^{K}\frac{\eta^k}{w^K}\langle \rho_i^*-\hat{\rho}_i^k, R_i^k \rangle=\sum_{k=1}^{K}\frac{\eta^k}{w^K}\langle q_i^*-\hat{q}_i^k, R_i^k \rangle.
	\end{align*}
	
	Notice that $q_i^*$ is induced by $\pi_i^*$ and $P_i$. Thus, we have $q_i^*\in \Delta_{i,\delta_i}(P_i)$ and $q_i^*\in \Delta_{i,\delta_i}(\CP_i^k)\ \forall k$.  From the OMD update rule, we have  
	\begin{align*}
		\hq^{k+1}_i=\argmax_{\hq_i\in \Delta_{i,\delta_i}(\CP_i^k)}\big\{\eta^k\langle \hq_i, R^k_i\rangle -D_{h_i}(\hq_i||\hq_i^k)\big\} = \argmax_{\hq_i\in \Delta_{i,\delta_i}(\CP_i^k)}\big\{\eta^k\langle \hq_i, R^k_i\rangle -h_i(\hq_i) +\la \nabla h_i(\hq_i^k),\hq_i-\hq_i^k\ra\big\},
	\end{align*}
	which by the first order optimality condition implies
	\begin{align*}
		\la \eta^k R_i^k+\nabla h_i(\hq_i^k)-\nabla h_i(\hq_i^{k+1}),q_i^*-\hq_i^{k+1}\ra \le 0.
	\end{align*}
	
	Therefore, we have
	\begin{align}
		\eta^k\la R_i^k,q_i^*-\hq_i^{k+1}\ra &\le \la -\nabla h_i(\hq_i^k)+\nabla h_i(\hq_i^{k+1}),q_i^*-\hq_i^{k+1}\ra\cr
		& \stackrel{\eqref{eq:Bregman_identity}}{=} D_{h_i}(q_i^*||\hq_i^k)-D_{h_i}(q_i^*||\hq_i^{k+1})-D_{h_i}(\hq_i^{k+1}||\hq_i^{k}).\label{eq:regret_1}
	\end{align}
	Since $h_i(\cdot)$ is $\mu$-strongly convex, we have $D_{h_i}(\hq_i^{k+1}||\hq_i^{k})\ge \frac{\mu}{2}\|\hq_i^{k+1}-\hq_i^{k}\|^2$. Also, we have $\eta^k\la R_i^k,\hq_i^k-\hq_i^{k+1}\ra \le \frac{\mu}{2}\|\hq_i^{k+1}-\hq_i^{k}\|^2+\frac{(\eta^k)^2}{2\mu}\|R_i^k\|^2$. Therefore, from \eqref{eq:regret_1} we have
	\begin{align}
		\eta^k\la R_i^k,q_i^*-\hq_i^{k}\ra&\le D_{h_i}(q_i^*||\hq_i^k)-D_{h_i}(q_i^*||\hq_i^{k+1})-D_{h_i}(\hq_i^{k+1}||\hq_i^{k}) +\eta^k\la R_i^k,\hq_i^k-\hq_i^{k+1}\ra\cr
		&\le D_{h_i}(q_i^*||\hq_i^k)-D_{h_i}(q_i^*||\hq_i^{k+1})-\frac{\mu}{2}\|\hq_i^{k+1}-\hq_i^{k}\|^2+\frac{\mu}{2}\|\hq_i^{k+1}-\hq_i^{k}\|^2+\frac{(\eta^k)^2}{2\mu}\|R_i^k\|^2\cr
		& = D_{h_i}(q_i^*||\hq_i^k)-D_{h_i}(q_i^*||\hq_i^{k+1})+\frac{(\eta^k)^2}{2\mu}\|R_i^k\|^2\label{eq:regret_step}.
	\end{align}
	Taking a telescoping sum of \eqref{eq:regret_step}, we obtain
	\begin{align*}
		\sum_{k=1}^{K}\eta^k\langle q_i^*-\hat{q}_i^k, R_i^k \rangle&\le \frac{1}{2\mu}\sum_{k=1}^{K}(\eta^k)^2\|R_i^k\|^2 +D_{h_i}(q_i^*||\hq_i^1)-D_{h_i}(q_i^*||\hq_i^{K+1})\\
		&\le \frac{1}{2\mu}\sum_{k=1}^{K}(\eta^k)^2\|R_i^k\|^2 +D_{h_i}(q_i^*||\hq_i^1)
	\end{align*}
	Since $R_i^k(s_i,a_i,s'_i)\in [0,1]\ \forall s_i,a_i,s'_i$ , we have $\|R_i^k\|^2\le |A_i||S_i|^2$. Thus, we can write
	\begin{align*}
		\text{Regret} &=\sum_{k=1}^{K}\frac{\eta^k}{w^K}\langle q_i^*-\hat{q}_i^k, R_i^k \rangle\\
		&\le \frac{|S_i|^2|A_i|}{2\mu}\sum_{k=1}^{K}\frac{(\eta^k)^2}{w^K}+\frac{D_{h_i}(q_i^*||\hq_i^1)}{w^K}.
	\end{align*}
\end{proof}

\bias*
\begin{proof}
	To bound the Bias, we fist define $G_i^k = \sum_{k=1}^{K}\frac{\eta^k}{w^K}\langle \rho_i^*-\hat{\rho}_i^k, \E_{k}[R_i^k]-R_i^k) \rangle$. Notice that $\{G_i^k\}_{k=1}^K$ is a zero-mean martingle such that for all $k \in [K]$,
	\begin{align*}
		|G_i^k-G_i^{k-1}| &= \frac{\eta^k}{w^K}\la \rho_i^*-\hat{\rho}_i^k, \E_{k}[R_i^k]-R_i^k) \ra\le \frac{\eta^k}{w^K}\|\rho_i^*-\hat{\rho}_i^k\|_1\le \frac{2\eta^k}{w^K},
	\end{align*}
Therefore, by Azuma's inequality, with probability at least $1-\gamma/n$, we have
	\begin{align*}
		G_i^K&\le \sqrt{2\ln\frac{n}{\delta}\sum_{k=1}^K\frac{4(\eta^k)^2}{(w^K)^2}}=\frac{2}{w^K}\sqrt{2\ln\frac{n}{\gamma}\sum_{k=1}^K(\eta^k)^2}.
	\end{align*}
	Now, using Lemma \ref{lem:Rbias}, we can bound the Bias as
	\begin{align*}
		\text{Bias} &= \sum_{k=1}^{K}\frac{\eta^k}{w^K}\langle \rho_i^*-\hat{\rho}_i^k, v_i(\rho_{-i}^k)-R_i^k) \rangle\\
		&=G_i^K + \sum_{k=1}^{K}\frac{\eta^k}{w^K}\langle \rho_i^*-\hat{\rho}_i^k, v_i(\rho_{-i}^k)-\E_{k}[R_i^k] \rangle\\
		&\le \frac{2}{w^K}\sqrt{2\ln\frac{n}{\gamma}\sum_{k=1}^K(\eta^k)^2} + \sum_{k=1}^{K}\frac{\eta^k}{w^K}e^{-\frac{d}{\tau}}\|\rho_i^*-\hat{\rho}_i^k\|_1\\
		&{\le}\frac{2}{w^K}\sqrt{2\ln\frac{n}{\gamma}\sum_{k=1}^K(\eta^k)^2} +2e^{-\frac{d}{\tau}}\\
	\end{align*}
 where the first inequality follows from \eqref{eq:Rbias_1}.
\end{proof}

\medskip
\section*{Appendix B: Proof of Lemmas For Theorem \ref{theo:assymptotic}}
\label{subsec:appendixb}

\medskip
\confidentb*
\begin{proof}
	The proof is almost identical to that of Lemma \ref{lem:confident1} and is omitted for the sake of brevity. 
\end{proof}

\qhatconverge*
\begin{proof}
	It suffices to show that 
	\begin{align*}
		\lim_{k\to\infty} q^k_i-\hq_i^k = \mathbf{0},\quad \forall i\in [n].
	\end{align*}
	In fact, from the choice of $\epsilon_i^k$, we have $\lim_{k\to\infty} \epsilon_i^k = \mathbf{0}\ \forall i$, such that $\lim_{k\to\infty} \hat{P}_i^k= P_i$. Then, using \eqref{eq:niu_1norm} in the proof of Lemma \ref{lem:error}, we have $\lim_{k\to\infty}\|\hat{\nu}_i^k-\nu_i^k\|_1 = 0$, and hence $\lim_{k\to\infty}\hat{\nu}_i^k-\nu_i^k = \mathbf{0}$ and $\lim_{k\to\infty}\hat{\rho}_i^k-\rho_i^k = \mathbf{0}$. Finally, from the fact that
	\begin{align*}
		q_i^k(s_i,a_i,s'_i) = \rho_i^k(s_i,a_i) P_i(s'_i|s_i,a_i),\quad \hat{q}_i^k(s_i,a_i,s'_i) = \hat{\rho}_i^k(s_i,a_i) \hat{P}_i^k(s'_i|s_i,a_i) \quad\forall s_i,a_i,s'_i,
	\end{align*}
	we obtain $\lim_{k\to\infty} q^k_i-\hq_i^k = \mathbf{0}\ \forall i\in [n]$.
\end{proof}

\recurrent*
\begin{proof}
	Let $U\subset \boldsymbol{\Delta}(\boldsymbol{P})$ be an open neighbourhood of $\boldsymbol{q}^*$ and assume to the contrary that, with positive probability, $\boldsymbol{\hq}^k\not\in U$ for all sufficiently large $k$.  By starting the sequence at a later index if necessary, we may assume that $\boldsymbol{\hq}^k\not\in U$ for all $k$ without loss of generality. Thus, by Assumption \ref{ass:g_stable}, there exists $c>0$ such that
	\begin{align}
		\label{eq:ge_c}
		\sum_{i=1}^n \la l_i(q_{-i}^k),q_i^*-q_i^k\ra \ge c\qquad \forall k \in \N.
	\end{align}
	Since $\bs{q}^*\in\bs{\Delta}_{\bs{\delta}}(\bs{P})$, conditioned on the event that $\{P_i \in \mathcal{P}_i^k\ \forall k \in \N, i \in [n]\}$, we have $q_i^*\in \Delta_{i,\delta_i}(\CP_i^k)\ \forall k\in \N$. Following the same derivation as in the proof of  Lemma \ref{lem:regret}, from \eqref{eq:regret_step} we have the following:
	\begin{align}
		D_{h_i}(q_i^*||\hq_i^{k+1})&\le D_{h_i}(q_i^*||\hq_i^{k})+\eta_k\la R_i^k,\hq^{k}_i-q_i^*\ra +\frac{(\eta^k)^2}{2\mu}\|R_i^k\|^2\cr
		&= D_{h_i}(q_i^*||\hq_i^{k+1})+\eta_k\la l_i^k(q_{-i}),q^{k}_i-q_i^*\ra +\frac{(\eta^k)^2}{2\mu}\|R_i^k\|^2+\eta_k\la l_i^k(q_{-i}),q^{k}_i-\hq^{k}_i\ra \cr
		&+\eta_k\la \E_k[R_i^k]-l_i^k(q_{-i}),\hq^{k}_i-q_i^*\ra +\eta_k\la R_i^k-\E_k[R_i^k],\hq^{k}_i-q_i^*\ra .\label{eq:recurrent_onestep}
	\end{align}
	Taking a telescoping sum of the above expression over $k =1,\ldots,K$, we obtain
	\begin{align*}
		D_{h_i}(q_i^*||\hq_i^{K+1})&\le D_{h_i}(q_i^*||\hq_i^{1})+\sum_{k=1}^K\eta^k\la l_i^k(q_{-i}),q^{k}_i-q_i^*\ra+\sum_{k=1}^K\frac{(\eta^k)^2}{2\mu}\|R_i^k\|^2+\sum_{k=1}^K\eta^k\la l_i^k(q_{-i}),q^{k}_i-\hq^{k}_i\ra \cr
		&+\sum_{k=1}^K\eta^k\la \E_k[R_i^k]-l_i^k(q_{-i}),\hq^{k}_i-q_i^*\ra +\sum_{k=1}^K\eta^k\la R_i^k-\E_k[R_i^k],\hq^{k}_i-q_i^*\ra . 
	\end{align*}
	Bu summing this relation over all $i\in [n]$ and using $w^K = \sum_{k=1}^K \eta^k$, we have
	\begin{align}
		&\boldsymbol{D}_{\boldsymbol{h}}(\boldsymbol{q}^*||\boldsymbol{\hq}^{K+1})\le \boldsymbol{D}_{\boldsymbol{h}}(\boldsymbol{q}^*||\boldsymbol{\hq}^1)
		+w^K\Big(\sum_{i=1}^n\sum_{k=1}^K\frac{\eta^k}{w^K}\la l_i^k(q_{-i}),q^{k}_i-q_i^*\ra
		+\underbrace{\sum_{i=1}^n\sum_{k=1}^K\frac{(\eta^k)^2}{2\mu w^K}\|R_i^k\|^2}_{T_1^K}\cr
		&+\underbrace{\sum_{i=1}^n\sum_{k=1}^K\frac{\eta^k}{w^K}\la l_i^k(q_{-i}),q^{k}_i-\hq^{k}_i\ra}_{T_2^K} \!+\!\underbrace{\sum_{i=1}^n\sum_{k=1}^K\frac{\eta^k}{w^K}\la \E_k[R_i^k]-l_i^k(q_{-i}),\hq^{k}_i-q_i^*\ra}_{T_3^K} \!+\!\underbrace{\sum_{i=1}^n\sum_{k=1}^K\frac{\eta^k}{w^K}\la R_i^k \!-\!\E_k[R_i^k],\hq^{k}_i\!-\!q_i^*\ra}_{T_4^K}\Big).\cr\label{eq:recurrent_overall}
	\end{align}
	To derive a contradiction, in the following we show that as $K\to \infty$, the left-hand side of \eqref{eq:recurrent_overall} remains nonnegative, while the right-hand side approaches $-\infty$, almost surely. In fact, from \eqref{eq:ge_c}, we can see that 
	$$\lim_{K\to \infty}\sum_{i=1}^n\sum_{k=1}^K\frac{\eta^k}{w^K}\la l_i^k(q_{-i}),q^{k}_i-q_i^*\ra\le -c\sum_{k=1}^{\infty}\frac{\eta^k}{w^K} = -c.$$ 
	We also have $\boldsymbol{D}_{\boldsymbol{h}}(\boldsymbol{q}^*||\boldsymbol{\hq}^1)< \infty$. Since $\lim_{K\to \infty} w^K = \infty$ by the step-size assumption, we only need to show that 
 the terms $T_1^K,T_2^K,T_3^K$ and $T_4^K$ all converges to $0$, almost surely.\\
	For $T_1^K$, since $\|R_i^k\|^2\le |A_i||S_i|^2$ and $\sum_{k=1}^{\infty}(\eta^k)^2 < \infty$, we have
	\begin{align*}
		\lim_{K\to \infty}T_1^K\le \lim_{K\to \infty}\frac{\sum_{i=1}^n|S_i|^2|A_i|}{2\mu w^K}\sum_{k=1}^{\infty}(\eta^k)^2=0.
	\end{align*}
	For $T_2^K$, following the same derivations as in the proof of Lemma \ref{lem:error}, we have
	\begin{align*}
		\sum_{i=1}^n\sum_{k=1}^K\eta^k\la l_i^k(q_{-i}),q^{k}_i-\hq^{k}_i\ra&= \sum_{i=1}^n\sum_{k=1}^K\eta^k \langle  v_i(\rho_{-i}^k),\hat{\rho}_i^k-\rho_i^k \rangle \\
		&\le \sum_{i=1}^n\sum_{k=1}^K\eta^k \|\hat{\rho}_i^k-\rho_i^k\|_1\\
		&= \sum_{i=1}^n\sum_{k=1}^K\eta^k \|\hat{\nu}_i^k-\nu_i^k\|_1\\
		&\stackrel{\eqref{eq:niu_1norm}}{\le }\sum_{i=1}^n\sum_{k=1}^K\eta^k\frac{2|S_i|}{1-e^{-1/\tau}}\max_{s_i,a_i,s'_i}\epsilon_i^k(s_i,a_i,s'_i)\\
		&\leq\sum_{i=1}^n\sum_{k=1}^K\eta^k\frac{|S_i|}{1-e^{-1/\tau}}\sqrt{\frac{2\ln \frac{2nk^2|A_i||S_i|^2}{\delta}}{k}}\\
		&=\sum_{i=1}^n\frac{|S_i|}{1-e^{-1/\tau}}\sum_{k=1}^K\eta^k\sqrt{\frac{2\ln \frac{2nk^2\max_i(|A_i||S_i|^2)}{\delta}}{k}}.
	\end{align*}
	Since $\eta^k$ satisfies $\sum_{k=1}^{\infty}\eta^k\sqrt{\frac{\ln k}{k}} < \infty $, so we have 
	\begin{align*}
		\lim_{K\to \infty}T_2^K\le \lim_{K\to \infty}\frac{1}{w^K}\Big(\sum_{i=1}^n\frac{|S_i|}{1-e^{-1/\tau}}\sum_{k=1}^{\infty}\eta^k\sqrt{\frac{2\ln \frac{2nk^2\max_i(|A_i||S_i|^2)}{\delta}}{k}} \Big)= 0.
	\end{align*}
	For $T_3^K$, using Lemma \ref{lem:Rbias}, we can write
	\begin{align*}
		\sum_{i=1}^n\sum_{k=1}^K\eta^k\la \E_k[R_i^k]-l_i^k(q_{-i}),\hq^{k}_i-q_i^*\ra&\le \sum_{i=1}^n\sum_{k=1}^K\eta^k e^{-\frac{d^k}{\tau}}\|\hq^{k}_i-q_i^*\|_1\le \sum_{i=1}^n\sum_{k=1}^K2\eta^k  e^{-\frac{d^k}{\tau}}.
	\end{align*}
	Therefore, from the choice of $d^k = 2\tau\ln k$, we have
	\begin{align*}
		\lim_{K\to \infty}T_3^K\le \lim_{K\to \infty}\frac{1}{w^K}(2n\sum_{k=1}^{\infty}\frac{\eta^k}{k^2})=0.
	\end{align*}
	Finally, for the term $T_4^K$, let $\boldsymbol{G}^k = \sum_{i=1}^n\la R_i^k-\E_k[R_i^k],\hq^{k}_i-q_i^*\ra$. Then, $\boldsymbol{G}^k $ is a martingale difference sequence with bounded $L_2$-norm, because 
	\begin{align*}
		\E_k[(\boldsymbol{G}^k)^2]& = \E_k [\sum_{i=1}^n\la R_i^k-\E_k[R_i^k],\hq^{k}_i-q_i^*\ra^2]\\
		&\le \E_k [\sum_{i=1}^n \|\hq^{k}_i-q_i^*\|_1^2]\le 4n < \infty.
	\end{align*}
	Therefore, using the strong law of large numbers for martingale difference sequences \cite[Theorem 2.18]{hall2014martingale}, we get
	\begin{align*}
		\lim_{K\to \infty}T_4^K = \lim_{K\to \infty}\sum_{k=1}^K\frac{\eta^k}{w^K}\boldsymbol{G}^k = 0 \quad \text{a.s.}
	\end{align*}
	Putting everything together, we have
	\begin{align*}
		0\le \lim_{K\to \infty}\boldsymbol{D}_{\boldsymbol{h}}(\boldsymbol{q}^*||\boldsymbol{\hq}^{K+1})\le -\infty, \quad \text{a.s.}
	\end{align*}
	a contradiction. This completes the proof.
\end{proof}

\stayin*
\begin{proof}
	For any $k_0$ and $k> k_0$, by following the same derivation as in the proof of Lemma \ref{lem:recurrent}, taking a telescoping sum of \eqref{eq:recurrent_onestep} over $r =k_0,\ldots,k-1$, and then summing over $i = 1,\ldots, n$, we obtain
	\begin{align}
		&\boldsymbol{D}_{\boldsymbol{h}}(\boldsymbol{q}^*||\boldsymbol{\hq}^{k})\le \boldsymbol{D}_{\boldsymbol{h}}(\boldsymbol{q}^*||\boldsymbol{\hq}^{k_0})
		+\sum_{i=1}^n\sum_{r=k_0}^{k-1}\eta^r\la l_i^r(q_{-i}),q^{r}_i-q_i^*\ra
		+\underbrace{\sum_{i=1}^n\sum_{r=k_0}^{k-1}\frac{(\eta^r)^2}{2\mu }\|R_i^r\|^2}_{U_1^k}\cr
		&+\underbrace{\sum_{i=1}^n\sum_{r=k_0}^{k-1}\eta^r\la l_i^r(q_{-i}),q^{r}_i-\hq^{r}_i\ra}_{U_2^k} \!+\!\underbrace{\sum_{i=1}^n\sum_{r=k_0}^{k-1}\eta^r\la \E_r[R_i^r]-l_i^r(q_{-i}),\hq^{r}_i-q_i^*\ra}_{U_3^k} \!+\!\underbrace{\sum_{i=1}^n\sum_{r=k_0}^{k-1}\eta^r\la R_i^r \!-\!\E_r[R_i^r],\hq^{r}_i\!-\!q_i^*\ra}_{U_4^k}.\cr\label{eq:recurrent_overall_2}
	\end{align}
	From Assumption \ref{ass:g_stable}, we know that $\sum_{i=1}^n\sum_{r=k_0}^{k-1}\eta^r\la l_i^r(q_{-i}),q^{r}_i-q_i^*\ra\le 0\ \forall k>k_0$. Therefore, we have 
	\begin{align}
		\P(\boldsymbol{D}_{\boldsymbol{h}}(\boldsymbol{q}^*||\boldsymbol{\hq}^k) \le \epsilon,\forall k>k_0)&\ge 1-\P(\boldsymbol{D}_{\boldsymbol{h}}(\boldsymbol{q}^*||\boldsymbol{\hq}^{k_0} >\epsilon/5)-\sum_{j=1}^4\P(\sup_{k>k_0} \{U_j^k\}>\epsilon/5).\label{eq:prob_whole}
	\end{align}
	Next, we give a high probability bound for each of the terms $\sup_{k>k_0} \{U_j^k\}$, $j = 1,2,3,4$.\\
	For $U_1^k$, we have $\|R_i^k\|^2\le |A_i||S_i|^2$, and since $\sum_{k=1}^{\infty}(\eta^k)^2 < \infty $, we get $\lim_{k\to \infty}\sum_{r=k}^{\infty}(\eta^r)^2 =0$. Therefore, there exists $k_1 \in \N$ such that $$\sum_{r=k_1}^{\infty}(\eta^r)^2 \le \frac{\epsilon}{5} \frac{1}{\sum_{i=1}^n\frac{|S_i|^2|A_i|}{2\mu}}.$$
	by choosing $k_0>k_1$, for any $k>k_0$, we have
	\begin{align*}
		\sum_{i=1}^n\sum_{r=k_0}^{k-1}\frac{(\eta^r)^2}{2\mu }\|R_i^r\|^2&\le \sum_{r=k_0}^{k-1}(\eta^r)^2\cdot \sum_{i=1}^n\frac{|S_i|^2|A_i|}{2\mu}\\
		&\le \sum_{r=k_1}^{\infty}(\eta^r)^2\cdot \sum_{i=1}^n\frac{|S_i|^2|A_i|}{2\mu}\\
		&\le \frac{\epsilon}{5}.
	\end{align*}
	Thus, for any $k_0\ge k_1$, we have 
	\begin{align}
		\P(\sup_{k>k_0} \{U_1^k\}>\epsilon/5) = 0.\label{eq:prob_U1}
	\end{align}
	For the term $U_2^k$, by following the same derivations as in the proof of Lemma \ref{lem:recurrent} (term $T_2^K$), we have
	\begin{align*}
		U_2^k\le \sum_{i=1}^n\frac{|S_i|}{1-e^{-1/\tau}}\sum_{r=k_0}^{k-1}\eta^r\sqrt{\frac{2\ln \frac{2nr^2\max_i(|A_i||S_i|^2)}{\delta}}{r}}.
	\end{align*}
	Since $\eta^k$ satisfies $\sum_{k=1}^{\infty}\eta^k\sqrt{\frac{\ln k}{k}} < \infty $, we have $\lim_{k\to \infty}\sum_{r=k}^{\infty}\eta^r\sqrt{\frac{\ln r}{r}}=0$. Similarly, there exists $k_2 \in \N$ such that $$\sum_{r=k_2}^{\infty}\eta^r\sqrt{\frac{2\ln \frac{2nr^2\max_i(|A_i||S_i|^2)}{\delta}}{r}} \le \frac{\epsilon}{5}\cdot \frac{1}{\sum_{i=1}^n\frac{|S_i|}{1-e^{-1/\tau}}}.$$
	Choosing $k_0>k_2$, we have for any $k>k_0$,
	\begin{align*}
		U_2^k&\le \sum_{i=1}^n\frac{|S_i|}{1-e^{-1/\tau}}\sum_{r=k_0}^{k-1}\eta^r\sqrt{\frac{2\ln \frac{2nr^2\max_i(|A_i||S_i|^2)}{\delta}}{r}}\\
		&\le \sum_{i=1}^n\frac{|S_i|}{1-e^{-1/\tau}}\sum_{r=k_2}^{\infty }\eta^r\sqrt{\frac{2\ln \frac{2nr^2\max_i(|S_i|^2|A_i|)}{\delta}}{r}}\\
		&\le \frac{\epsilon}{5}.
	\end{align*}
	Therefore, for any $k_0\ge k_2$, we have 
	\begin{align}
		\P(\sup_{k>k_0} \{U_2^k\}>\epsilon/5) = 0.\label{eq:prob_U2}
	\end{align}
	For the term $U_3^k$, by following the same derivations as in the proof of Lemma \ref{lem:recurrent} (term $T_3^K$) and the choice of $d^k = 2\tau\ln k$, we can write
	\begin{align*}
		\sum_{i=1}^n\sum_{r=k_0}^{k-1}\eta^r\la \E_r[R_i^r]-l_i^r(q_{-i}),\hq^{r}_i-q_i^*\ra\le \sum_{i=1}^n\sum_{r=k_0}^{k-1}\eta^r  e^{-\frac{d^r}{\tau}}\cdot 2\le 2n\sum_{r=k_0}^{k-1}\frac{\eta^r}{r^2}.
	\end{align*}
	This shows that there exists $k_3 \in \N$ such that $$\sum_{r=k_3}^{\infty}\frac{\eta^r }{r^2}\le \frac{\epsilon}{5} \frac{1}{2n}.$$
	By choosing $k_0>k_3$, for any $k>k_0$, we have
	\begin{align*}
		U_3^k\le 2n\sum_{r=k_0}^{k-1}\sum_{r=k_0}^{k-1}\frac{\eta^r}{r^2}\le 2n\sum_{r=k_0}^{k-1}\sum_{r=k_3}^{\infty}\frac{\eta^r}{r^2}\le \frac{\epsilon}{5}.
	\end{align*}
	Therefore, for any $k_0\ge k_3$, we have 
	\begin{align}
		\P(\sup_{k>k_0} \{U_3^k\}>\epsilon/5) = 0.\label{eq:prob_U3}
	\end{align}
	For the term $U_4^k$, and as in the proof of Lemma \ref{lem:recurrent} (term $T_4^K$), let $\boldsymbol{G}^k = \sum_{i=1}^n\la R_i^k-\E_k[R_i^k],\hq^{k}_i-q_i^*\ra$. Then, $\boldsymbol{G}^k $ is a martingale difference sequence such that $\E_k[|\boldsymbol{G}^k|^2]\le 4n$. Therefore, $U_4^k= \sum_{r = k_0}^{k-1}\eta^r\boldsymbol{G}^r$ is a zero-mean martingale sequence with
	\begin{align*}
		\E[|U_4^k|^2]\le \sum_{r = k_0}^{k-1}(\eta^r)^2\E_r[|\boldsymbol{G}^r|^2]\le 4n\sum_{r = k_0}^{k-1}(\eta^r)^2.
	\end{align*}
	Let $F^k$ denote the event that $\{\sup_{k_0<r\le k}U_4^r> \frac{\epsilon}{5}\}$. Using Doob’s maximal inequality for martingales \cite[Theorem 2.1]{hall2014martingale}, we have
	\begin{align*}
		\P(F^k)\le \frac{25\E[|U_4^k|^2]}{\epsilon^2}\le \frac{100n\sum_{r = k_0}^{k-1}(\eta^r)^2}{\epsilon^2}.
	\end{align*}
	Since $\{\sup_{k>k_0} U_3^k>\epsilon/5\} = \cup_{k = k_0+1}^{\infty}F^k$ and $F^k \subseteq F^{k+1},\forall k>k_0$, we have
	\begin{align*}
		\P(\sup_{k>k_0} \{U_4^k\}>\epsilon/5) = \lim_{k\to \infty}\P(F^k)\le \frac{100n\sum_{r = k_0}^{\infty}(\eta^r)^2}{\epsilon^2}.
	\end{align*}
	Moreover, using the choice of step-sizes, $\lim_{k\to \infty}\sum_{r=k}^{\infty}(\eta^r)^2 =0$. Thus, there exists $k_4 \in \N$ such that $$\sum_{r=k_1}^{\infty}(\eta^r)^2 \le \delta  \frac{\epsilon^2}{100n}.$$
	Therefore, for any $k_0\ge k_4$, we have
	\begin{align}
		\P(\sup_{k>k_0} \{U_4^k\}>\epsilon/5)\le \frac{100n\sum_{r = k_0}^{\infty}(\eta^r)^2}{\epsilon^2}\le \frac{100n\sum_{r = k_4}^{\infty}(\eta^r)^2}{\epsilon^2}\le \delta.\label{eq:prob_U4}
	\end{align}
	Finally, from Lemmas \ref{lem:recurrent} and \ref{lem:qhat_converge}, we can find $k_0\ge \max\{k_1,k_2,k_3,k_4\}$ such that $\boldsymbol{D}_{\boldsymbol{h}}(\boldsymbol{q}^*||\boldsymbol{\hq}^{k_0})\le \frac{\epsilon}{5}$, and hence 
	\begin{align}
		\P(\boldsymbol{D}_{\boldsymbol{h}}(\boldsymbol{q}^*||\boldsymbol{\hq}^{k_0} >\epsilon/5) = 0.\label{eq:prob_U0}
	\end{align}
	Combining relations \eqref{eq:prob_whole}, \eqref{eq:prob_U1}, \eqref{eq:prob_U2}, \eqref{eq:prob_U3}, \eqref{eq:prob_U4}, and \eqref{eq:prob_U0}, completes the proof.
\end{proof}

\end{document}